\documentclass[reqno,12pt,letterpaper]{amsart}
\usepackage[dvips]{epsfig}
\usepackage{graphics}
\usepackage{latexsym}
\usepackage{verbatim}
\usepackage{amsmath}
\usepackage{amsthm}
\usepackage{amssymb}
\usepackage[utf8]{inputenc}
\usepackage[english]{babel}
 \usepackage{hyperref}
\newtheorem{theorem}{Theorem}[section]
\newtheorem{lemma}[theorem]{Lemma}

\newtheorem{proposition}[theorem]{Proposition}

\newtheorem{remark}[theorem]{Remark}

\newcommand{\C}{\mathbb{C}}

\usepackage{bbm}
\usepackage[utf8]{inputenc}
\usepackage[english]{babel}
\usepackage{mathtools}
\usepackage{mleftright}

\makeatletter
\def\Ddots{\mathinner{\mkern1mu\raise\p@
\vbox{\kern7\p@\hbox{.}}\mkern2mu
\raise4\p@\hbox{.}\mkern2mu\raise7\p@\hbox{.}\mkern1mu}}
\makeatother

\newcommand\scalemath[2]{\scalebox{#1}{\mbox{\ensuremath{\displaystyle #2}}}}

\usepackage[usenames, dvipsnames]{color}
\usepackage[all,cmtip]{xy}
\usepackage{graphicx}
\usepackage{caption}
\usepackage{blindtext}

\usepackage{amsmath}

\graphicspath{ {./images/} }
\begin{document}

\title[A perturbation of the harmonic chain of coupled oscillators]{Quantitative Rates of Convergence to Non-Equilibrium Steady State for a Weakly Anharmonic Chain of Oscillators}

\author{ANGELIKI MENEGAKI}
\email{angeliki.menegaki@dpmms.cam.ac.uk}
\address{DPMMS, University of Cambridge, Wilberforce Rd, Cambridge CB3 0WA, UK}

\maketitle

\begin{abstract}
We study a $1$-dimensional chain of $N$ weakly anharmonic classical oscillators coupled at its ends to heat baths at different temperatures.  Each oscillator is subject to pinning potential and it also interacts with its nearest neighbors.  In our set up both potentials are homogeneous and bounded (with $N$ dependent bounds) perturbations of the harmonic ones. We show how a generalized version of Bakry-Emery theory can be adapted to this case of a hypoelliptic generator which is inspired by F. Baudoin (2017).
 By that we prove exponential convergence to non-equilibrium steady state (NESS) in Wasserstein-Kantorovich distance and in relative entropy with quantitative rates. We estimate the constants in the rate by solving a Lyapunov-type matrix equation and we obtain that the exponential rate, for the homogeneous chain, has order bigger than $N^{-3}$.  For the purely harmonic chain the order of the rate is in $ [N^{-3},N^{-1}]$. This shows that, in this set up, the spectral gap decays at most polynomially with $N$.
\end{abstract}

\setcounter{secnumdepth}{7}
\setcounter{tocdepth}{1}
\tableofcontents

\section{Introduction} \label{introduction}
 
\subsection{Description of the model} 

We consider a model for heat conduction consisting of a one-dimensional chain of $N$ coupled oscillators whose dynamics governed by the Hamiltonian $$ H(p,q)= \sum_{1\leq i \leq N} \left( \frac{p_i^2}{2} + V(q_i) \right) + \sum_{i=1}^{N-1}U(q_{i+1}-q_i) $$ where $q_i \in \mathbb{R}$ is the displacements of the atoms with respect to their respective equilibrium positions and $p_i \in \mathbb{R}$ the momentum of the atoms $i=1,\dots,N$ in the phase space $  \mathbb{R}^N \times \mathbb{R}^N$. Each particle has its own pinning potential $V$ and it also interacts with its nearest neighbors through an interaction potential $U$. Notice that here we take all the masses $m_i=1$ for simplicity. The classical Hamiltonian dynamics is perturbed by noise and friction in the following way: the two ends of the chain are in contact with heat baths at two different temperatures $T_L$, $T_R >0 $. In this case we consider the interaction of the chain with the heat baths to be described by two Langevin processes.  So our dynamics is described by the following system of SDEs: \begin{equation} 
\begin{split} 
\label{eq: SDE}
 \mbox{d}q_i(t)&=p_i(t) \mbox{d}t  \quad \text{for} \quad i=1,\dots,N, \\
 \mbox{d}p_i(t)&= ( - \partial_{q_i} H) \mbox{d}t \quad \text{for} \quad  i=2,\dots,N-1, \\ 
 \mbox{d}p_1(t)&=(-\partial_{q_1}H - \gamma_1 p_1)\mbox{d}t + \sqrt{2 \gamma_1 T_1} \mbox{d}W_1(t),\\
 \mbox{d}p_N(t)&=(-\partial_{q_N}H - \gamma_N p_N)\mbox{d}t + \sqrt{2 \gamma_N T_N} \mbox{d}W_i(t) 
 \end{split}
\end{equation}
where $\gamma_i$ are the friction constants, $T_i$ the two temperatures and $W_1,W_N$ are two independent normalized Wiener processes.\\

\noindent
The dynamics \eqref{eq: SDE} is equivalently described by Liouville equation \begin{equation}\label{Liouville}
\partial_tf = \mathcal{L}f\quad \text{with}\quad f(0,p,q)=f_0(q,p) 
\end{equation} 
 where $\mathcal{L}$ is the second order differential operator
 \begin{equation} \label{eq:generator}
\mathcal{L}= \sum_{i=1}^N ( p_i \partial_{q_i}- \partial_{q_i}H \partial_{p_i})- \gamma_1 p_1\partial_{p_1}- \gamma_N p_N\partial_{p_N}+\gamma_1T_L \partial_{p_1}^2+\gamma_NT_R \partial_{p_N}^2.
\end{equation}
 
 \subsubsection{Derivation of $\mathcal{L}$}
 Since the solution to \eqref{eq: SDE} form a Markov process, we define the transition probabilities $$ P_t^*(z ,dy)=\mathbb{P}( z_t \in \mbox{d}y |z_0=z)\quad \text{with} \quad \int_{\mathbb{R}^{2N}} P_t^*(z, \mbox{d}y)=1, \quad \forall\ z \in \mathbb{R}^{2N} $$
where $z$ is the initial condition and $ P_t^*$ satisfies the Chapman-Kolmogorov relation $$  P_{t+s}^*(z,\mbox{d}y)=\int_{w \in \mathbb{R}^{2N}}P_t^*(z , \mbox{d}w)P_s^* (w , \mbox{d}y).$$ Thus we consider a semigroup $\{P_t^*$ , $ t \geq 0\}$ on the space of Borel probability measures on the space $\mathbb{R}^{2N}$ such that $$ (P_t^* \mu)(B) = \int_{\mathbb{R}^N \times \mathbb{R}^N} P_t^* (x,B)d\mu(x),$$ for $B$ a Borel subset on $\mathbb{R}^{2N}$.  
Now, one can similarly consider the \textit{dual} semigroup $ \{ P_t, t \geq 0 \}$ acting on observables. For any measurable function $f: \Omega \rightarrow \mathbb{R}$ we define $$ P_t f(z) = \int_{\mathbb{R}^{2N} } f(y)P_t^*(z,dy) = \mathbb{E}_z(f(z_t))$$ where $z_t=(p_t,q_t)$ solves \eqref{eq: SDE} and $\mathbb{E}_z\big( f(z_t) \big)$ is the expectation taking over all the realizations of the Brownian motion starting from $z \in \mathbb{R}^{2N}$.\\
Having a well-defined semigroup $\{P_t\}_{t\geq0}$  with invariant measure $\mu$, we can make sense of the following definition of the \textit{generator} of the semigroup $\{P_t\}_{t\geq0}$ (see for instance Section 1.4 in \cite{BGL14}) :  \begin{align}
\mathcal{L}f := \lim_{t\to 0^+} \frac{P_t f-f}{t}
\end{align}
for every $f \in C_c^{\infty}(\mathbb{R}^{2N})$.\\
 By applying Ito's formula, we get the expression of the generator of the Markov process which solves \eqref{eq: SDE} to be \eqref{eq:generator}.

\subsubsection{State of the art} This model was first used to describe heat diffusion and derive rigorously Fourier's law (for an overview see \cite{BLR00} and references therein).  Since then, it has been the subject of many studies, both from a numerical and from a theoretical perspective. Firstly, the purely harmonic case with several idealized reservoirs at different temperatures has been solved explicitly in \cite{RLL67} where they found exactly how the non-equilibrium stationary state (NESS) looks like: they show that the stationary measure is Gaussian in the positions and momenta of the system, corresponding to the distribution when the temperatures of all the reservoirs are equal, with the difference that the covariances between momenta and positions are non zero, but proportional to the temperature difference. For the anharmonic chain there are no explicit results in general. However it has been studied numerically for many different potentials and many kinds of heat baths, including the Langevin heat baths that we consider here. See for instance
\cite{MR2265846} \cite{GLPV00,LLP03} and references therein. \\

There are two facts in this model that make its rigorous study very challenging: first of all, we do not know explicitly the form of the invariant measure of \eqref{eq: SDE} and also our generator is highly degenerate, having the dissipation and noise acting only on two variables of momenta at the end of the chain. It is not difficult to see, though, that in the equilibrium case, \textit{i.e}. when the two temperatures are equal $T_L=T_R=T= \beta^{-1}$, the stationary measure is the Gibbs-Boltzmann measure $\mbox{d}\mu(p,q)=\exp(-\beta H(p,q)) \mbox{d}p\mbox{d}q$: after explicit calculations we have $\mathcal{L}^* e^{-\beta H(p,q)}=0$ since $$ \mathcal{L}^* = p \cdot \nabla_q + \nabla_q H \cdot \nabla_p + \big(  \gamma_i(1+p_i \partial_{p_i}) + \gamma_i T_i \partial_{p_i}^2\big)\delta_{i=1,i=N} $$ on $ L^2(\mathbb{R}^{2N}).$  So if $f_t(p,q)$ solves 
\eqref{Liouville} then $\int f_t(p,q) \mbox{d}\mu(p,q)$ does not depend on $t$. \\

Since we are interested in the theoretical aspects of the model, we refer to \cite{EPR99a, EPR99b}, which is the first rigorous studies of the anharmonic case. The existence of a steady state has only been obtained in some cases where the potentials act like polynomials near infinity. In particular making the following assumptions on pinning and interaction potentials $$ \lim_{\lambda \to \infty} \lambda^{-k} U(\lambda q )=a_k |q|^k\ \text{and}\ \lim_{\lambda \to \infty} \lambda^{1-k} U'(\lambda q )=k a_k |q|^{k-1} \text{sign}(q) $$ for constants $a_k>0$, 
and assuming that  \textit{the interaction potential is at least as strong as the pinning}, the existence and uniqueness of an invariant measure was first proved in \cite{EPR99a} using functional analytic methods: they show that the resolvent of the generator of \eqref{eq: SDE} is compact in a suitable weighted $L^2$ space. Later it was proven in \cite{RBT02} that the rate of convergence to equilibrium is exponential using probabilistic tools. \\ 

In these papers  they model the heat baths  slightly differently and a bit more complicated than the Langevin thermostats, with physical interpretation: the model of the reservoirs is the classical field theory given by linear wave equations with initial conditions distributed with respect to appropriate Gibbs measures at different temperatures, see also \cite[Section 2]{RB06}.  Later, an adaptation of a very similar probabilistic proof was provided in \cite{Car07} for the Langevin thermostats. The difference  with the Langevin heat baths is that the dissipation and the noise act on the momenta only indirectly through some auxialiary variables. \\ 

Regarding the existence, uniqueness of a Non-Equilibrium Stationary State and exponential convergence towards it in more complicated networks of oscillators (multi-dimensional cases) see \cite{CEHRB18}. The proofs there are inspired by the abovementioned works in the $1$-dimensional chains. \\
%Note, however, that in all these papers the dependence of the convergence rate on the number of the oscillators is not clear. \\

There are also cases where there is no convergence to equilibrium, when for instance $l > k$, \textit{i.e}. when the pinning is stronger than the coupling potential, see for example \cite{Hair09,HM09} where they cover some cases like that. In \cite{HM09} they show that the resolvent of the generator fails to be compact or/and that there is lack of spectral gap. In particular they show that with harmonic interaction, $0$ belongs in the essential spectrum of the generator as soon as the pinning potential is of the form $| q |^{k}$ for $k>3$. The conjecture is that this is true as soon as $k>\frac{2n}{2n-1}$ if $n$ is the center of the chain.

\subsection{Set up and main results} 
We consider a small perturbation of the harmonic chain, \textit{weakly anharmonic}: 
First, regarding the boundary conditions, we model the oscillators chain with \textit{rigidly fixed edges} meaning that we consider a chain beginning with an oscillator labelled $0$ and ending with one labelled $N+1$ under the hypothesis that $q_0=q_{N+1}=0$. The first and the last particle are pinned with additional harmonic forces, modelling their attachment to a wall.  These boundary conditions and heat baths modelled by adding two Ornstein-Uhlembeck processes at both ends as explained above, is the same model as in \cite{RLL67} and  is known as the \textit{Casher-Lebowitz model}, since it is also one of the models considered in \cite{CL71} in order to study the heat flux behaviour in disordered harmonic chains \footnote{The other one considered for studying the $N$-dependence of the energy flux in disordered harmonic chains, was first introduced by Rubin-Greer, \cite{RG71}, where the heat baths are semi-infinite chains distributed according to Gibbs equilibrium measures of temperatures $T_L, T_R$ (free boundaries). } .\\

Second, we assume that both pinning and interaction potentials of the oscillator chain differ from the harmonic ones by potentials with Hessians bounded from above by positive constants, uniformly in $q \in \mathbb{R}^N$. Eventually  we have a Hamiltonian of the form \begin{align} \label{Hamiltonian} H(p,q)= \sum_{i=1}^N \left( \frac{p_i^2}{2} + a \frac{ q_i^2}{2} + U_{\text{pin}}(q_i) \right) + \sum_{i=1}^{N-1}& \left( c \frac{(q_{i+1}-q_i)^2}{2} + U_{\text{int}}( q_{i+1}-q_i) \right)+  \\&+ \frac{c q_1^2}{2}+\frac{c q_N^2}{2} \notag
\end{align}
with \begin{align} \label{assumptions on the potentials} \sup_{(p,q) \in \mathbb{R}^{2N}}\|\text{Hess}\ U_{\text{pin}}(p,q) \|_{2}  \leq C_{pin}(N) \quad \text{and}\quad  \sup_{(p,q) \in \mathbb{R}^{2N}} \| \text{Hess}\ U_{\text{int}} (p,q)\|_{2} \leq C_{int}(N) \end{align}
where the positive constants $C_{pin}(N)$, $C_{int}(N)$  will be chosen later and they will depend on $N$. \\ % Then  we can choose, for instance,  $U_{\text{int}}$ such that $\nabla_q U_{\text{int}}(q) \sim -a(1+|q|)^{a-1}  -q \epsilon_{int} $ as $|q| \to \infty$ with $ 1 \leq a \leq 2$.\\

Denoting by $\mathcal{L}$ the infinitesimal generator of the Markov semigroup $(P_t)_{t\geq 0}$ which is defined as before, we look at Liouville equation $ \partial_tf = \mathcal{L}f$,  where the generator of the dynamics now is \begin{align*}  \mathcal{L} = p\cdot \nabla_q - &q B  \cdot\nabla_p - \nabla_q U_{\text{pin}} \cdot \nabla_p- \gamma p_1 \partial_{p_1} -\gamma p_N \partial_{p_N} + \gamma T_L \partial_{p_1}^2 + \gamma T_R \partial_{p_N}^2- \\&- \sum_{i=2}^{N-1} \Big( \nabla_q U_{\text{int}}(q_{i+1}-q_i) \cdot \nabla_p + \nabla_q U_{\text{int}}(q_{i}-q_{i-1}) \cdot \nabla_p \Big)
\end{align*} where we take all the friction constants equal $\gamma_1=\gamma_N=\gamma$, for the two temperatures $T_L,T_R$ we assume that they satisfy $T_L=T+\Delta T$, $T_R=T- \Delta T$, for some temperature difference $\Delta T >0$. Also, $B$ is the symmetric tridiagonal (Jacobi) matrix %so that for a vector $q \in \mathbb{R}^N$ :
%\begin{align}  \label{matrix B}
%(qB)_{i}:= -c (q_{i+1} + q_{i-1} ) + (a+2c) q_i,\quad  1 \le i \le N 
%\end{align} with $q_0=q_{N+1}=0$. 

\begin{align} \label{matrix B} B:= \begin{bmatrix} 
(a+2c) & -c &0 &0& \dots &0& 0 &0 \\ -c& (a +2c)  & -c &0& \cdots & 0 & 0&0 \\ 0 & -c & (a +2c)  & -c& \dots& 0 &0&0\\ \vdots &\vdots & \quad & \quad&\quad & \quad & \quad\\ \vdots &\vdots & \quad & \quad&\ddots & \ddots & \ddots\\ 0&0&0&0&\quad & -c & (a+2c)  & -c \\ 0 &0& 0&0 & \dots & 0& -c & (a+2c) 
\end{bmatrix}. \end{align}

It is convenient to see the above form of the generator in the following block-matrix form:\begin{equation} \label{generator} \mathcal{L} =- zM \cdot \nabla_z -\nabla_q \Phi(q) \cdot \nabla_p + \nabla_p \cdot \Gamma \Theta \nabla_p
\end{equation}
where $z=(p,q) \in \mathbb{R}^{2N}$, $\Phi(q)$ corresponds to the perturbing potentials so that $$  \Phi(q) =  U_{\text{pin}} + \sum_{i=2}^{N-1} \Big( U_{\text{int}}(q_{i+1}-q_i)  + U_{\text{int}}(q_{i}-q_{i-1}) \Big) ,$$ the matrix $\Gamma$ is the friction matrix $$ \Gamma = \text{diag}(\gamma,0 ,\cdots ,0, \gamma) $$ the matrix $\Theta$ is the temperature matrix $$ \Theta= \text{diag}(T_L,0,\cdots, 0,T_R)$$ and $M$ in blocks is the following \begin{equation} \label{matrix M} M = \begin{bmatrix} \Gamma & -I\\ B & 0 
\end{bmatrix}
\end{equation} where  $I$ is the identity matrix, so that it corresponds to the transport part of the operator, while $B$ and $\Gamma$ correspond to the harmonic part of the potentials and the drift from both ends, respectively.\\

\noindent
\textit{Motivation}. This study is motivated by a discussion opened in C. Villani's memoir on hypocoercivity, see Section 9.2 in \cite{Villani09}, about approaching questions on this heat conduction  model by hypocoercive techniques. This chain of coupled oscillators is a hypocoercive situation, where the diffusion only on the ends of the chain leads to a convergence to the stationary distribution exponentially fast, under the following assumptions on the potentials: strict convexity on interaction potential (being stronger than the pinning one) and bounded Hessians for both potentials. In particular, he points out that it might be possible to recover the previous results of exponential convergence in the weighted $H^1(\mu)$-norm for this different class of potentials (than the potentials assumed in \cite{EPR99b} for instance)  by applying a generalized version of Theorem 24 in \cite{Villani09}. For that, one needs to know some properties of the, non-explicit,  non-equilibrium steady state $\mu$: for instance, if it satisfies a Poincar\'{e} inequality or if the Hessian of the logarithm of its density is bounded.  
\\

\noindent
\textit{Main results}. Here, considering a perturbation of the harmonic chain (homogeneous case), instead we follow an approach that combines hypocoercivity  and Bakry-\'{E}mery theory of $\Gamma$ calculus and curvature conditions as in \cite{BakEm83}. We obtain the same results with Bakry-\'{E}mery but in a perturbed setting. This is explained in more details and is implemented in \hyperref[section with functional ineqaualities]{Section 3}.  The whole idea was inspired by F. Baudoin in \cite{Bau17}, where he used this combination in order to show exponential convergence to equilibrium for the Kinetic Fokker-Planck equation in $H^1$-norm and in Kantorovich-Wasserstein distance. 

By that it is possible to show, for the dynamics \eqref{eq: SDE} as well, exponential convergence in Kantorovich-Wasserstein distance and in relative entropy and to get quantitative rates of convergence in these distances, \textit{i.e.} to obtain information on the $N$-dependence of the rate.  In particular our estimates show that the convergence rate in the harmonic chain approach $0$ as $N$ tends to infinity at a \textit{polynomial} rate of order $1/N^{\gamma}$ where $ 1 \leq \gamma \leq 3$ and that it is bigger than $N^{-3}$ in the weakly anharmonic chain. 

In order to quantify the above rates, we estimate $\|b_N\|_2$, where $b_N$ is a block matrix defined in \hyperref[section with functional ineqaualities]{Section 3} as a solution of a matrix equation, \eqref{tupos gia b'}. Since $\|b_N\|_2$ appears in the rates in the Theorems \ref{maintheorem}, \ref{converg in entropy}, we start by stating this result:

\begin{proposition} \label{propos of induction} Considering the homogeneous scenario of the oscillators chain, with pinning coefficient $a>0$ and interaction coefficient $c>0$, there exists a  symmetric block matrix  $$b_{N}= \begin{bmatrix}  x_{N} & z_{N} \\ z_{N}^T & y_{N} 
\end{bmatrix} \in \mathbb{C}^{2N \times 2N}$$ that can be constructed as the unique solution of the following Lyapunov equation \begin{equation} \label{tupos gia b_N} 
b_{N} M +M^T b_{N} = \Pi_{N}
\end{equation} 
where $\Pi_{N} =   \operatorname{diag}(2T_L, 1, \dots,1,2T_R, 1, 1, \dots, 1,1)$ and $M= \begin{bmatrix} \Gamma & -I \\ B& 0 \end{bmatrix}$,  and the spectral norm of which, $ \| b_{N} \|_2, $ is bounded from above in terms of the dimension $N$, as  $\mathcal{O}(N^{3})$. 
\end{proposition}

\begin{theorem} \label{maintheorem}
We consider a chain of coupled oscillators whose dynamics are described by the system \eqref{eq: SDE} and the Hamiltonian is given by \eqref{Hamiltonian} under the assumptions on the potentials given by \eqref{assumptions on the potentials}. In particular we assume bounded perturbations of the harmonic chain with bounds depending on $N$, \eqref{behaviour of bounds}, so that $$ \|b_N\|_2^{-1} - 2 \big( C_{pin}(N)+C_{int}(N) \big)\|b_N\|_2 >0  . $$
 For a fixed number of particles $N$, we have a convergence to NESS in Wasserstein-Kantorovich distance \begin{equation} \label{convergence result} 
W_2(P_t^*f_0^1, P_t^* f_0^2) \leq \| b_N \|_{2}^{1/2}  \| b_N^{-1} \|_{2}^{1/2} e ^{- \lambda_N t} W_2 (f_0^1,f_0^2)
\end{equation} for $f_0^1,f_0^2$ initial data of the evolution equation. Here
  $b_N$ is the matrix that solves the equation \eqref{tupos gia b'} and $\lambda_N$ a function of $\|b_N\|_2$ as in \eqref{lambda_N}: $$ \lambda_N =  \| b_N \|_{2}^{-1}  - 2(C_{pin}(N) + C_{int}(N) ) \|b_N\|_2 \| b_N^{-1}\|_2.$$ 
Moreover, there is a unique stationary solution $f_{\infty}$, since all the solutions $f_t$ will converge towards it if we make the choice  $f_0^2= f_{\infty}$. Estimates on $\|b_N\|_2$ regarding $N$ are given by the Proposition \ref{propos of induction} and this allows to conclude that \begin{align} \label{quantit convergen in W2} W_2(P_t^*f_0^1, P_t^* f_0^2) \leq N^{\frac{3}{2}} e^{- \frac{\lambda_0}{N^3}t}\  W_2 (f_0^1,f_0^2) \end{align} and that the bounds on the perturbing potentials vanish asymptotically with rate $$   C_{pin}(N)+C_{int}(N)  \lesssim \frac{C_0}{N^6}$$ with constants $C_0, \lambda_0$ independent of $N$. 
 \end{theorem}
 
%\begin{remark} For the study of convergence in time for large particle systems, it is natural to work with the so-called extensive functionals, \textit{i.e.} that are additive (proportional) in $N$, examples of which are the Wasserstein metric or the relative entropy. So the normalized Wasserstein distance is tensorized correctly and thus we can not do better for the constant in front of the exponential in \eqref{quantit convergen in W2} by working with the $L^2$-norm for example, which can imply geometric dependence on $N$: one then has to wait even more time in terms of $N$ before the convergence to NESS starts taking place.
%The $L^2$-norm can imply geometric dependence in N that will destroy eventually the convergence in time: one then has to wait time proportional to $N$ before the convergence in NESS starts taking place.
 %For more details (in the context of Kac's model), see the discussion just before the subsection 1.5 in \cite{MM13}.
 %\end{remark}

Moreover, in the set up of Theorem \ref{maintheorem}, we get some qualitative information about the non-equilibrium steady distribution, like the validity of a Poincar\'{e} inequality and even better, a Log-Sobolev inequality:

\begin{proposition}[Log-Sobolev inequality] \label{log sobolev inqlt} Let $\mathcal{L} $ be the generator of the dynamics described by the SDEs \eqref{eq: SDE}. Let $\Gamma$ be the Carr\'{e} du Champ operator defined in \eqref{Gamma operator}, while $\mathcal{T}$ the perturbed quadratic form defined in \eqref{perturbed Gamma operator}.
Assuming that we have a gradient estimate of the form \eqref{strong gradient bound}, we obtain that the unique invariant measure $\mu=f_{\infty} $ from the Theorem \ref{maintheorem} satisfies a Log-Sobolev inequality $(LSI(C_N) )$ : 
\begin{equation} \label{LSI} \int_{\mathbb{R}^{2N}} f \log f\ d\mu - \int_{\mathbb{R}^{2N}} f\ d\mu\ \log \left( \int_{\mathbb{R}^{2N}} f\ d\mu \right) \leq C_N \int_{\mathbb{R}^{2N}} \frac{\mathcal{T}(f,f)}{f} d\mu.
\end{equation}
where $$C_N=\frac{ T_L\| b_N^{-1} \|_{2}}{2\lambda_N}.$$
\end{proposition}

 Consequently we have convergence to NESS in Entropy.
 Let us first define the following information-theoretical functionals: the Boltzmann $H$ functional \begin{equation} \label{Boltzmann functional}
H_{\mu}(\nu)=\int_{\mathbb{R}^{2N}} h \log h\ d\mu,\ \nu=h \mu
\end{equation}
and the relative Fisher information
\begin{equation}  \label{Fisher info}
I_{\mu}(\nu)= \int_{\mathbb{R}^{2N}} \frac{| \nabla h |^2}{h} d \mu,\ \nu=h \mu.
\end{equation}
Defining, for $f \in \C^{\infty}(\mathbb{R}^{2N})$, the functional $$\mathcal{E}(f) := \int_{\mathbb{R}^{2N}} \left( f \log f+ f\  \mathcal{T}( \log f, \log f)  \right) d\mu,$$ we have entropic convergence in the following sense, as in \cite[Section 6]{Villani09}:

 \begin{theorem}\label{converg in entropy} We consider a chain of coupled oscillators whose dynamics are described by the system \eqref{eq: SDE} and the Hamiltonian is given by \eqref{Hamiltonian} under the assumptions on the potentials given by \eqref{assumptions on the potentials}: we assume bounded perturbations of the harmonic chain with bounds depending on $N$, \eqref{behaviour of bounds}, so that $$ \|b_N\|_2^{-1} - 2 \big( C_{pin}(N)+C_{int}(N) \big)\|b_N\|_2 >0  . $$  Here
  $b_N$ is the matrix that solves the equation \eqref{tupos gia b'}. For a fixed number of particles $N$,  assuming that (i) $\mu$ is the invariant measure for $P_t$ and (ii) that it satisfies a Log-Sobolev inequality with constant $C_N>0$,  for all $f$ with $\mathcal{E}(f)< \infty $ i.e. the initial data have finite relative entropy with respect to $\mu$, we have a convergence to NESS in Entropy: \begin{equation}
\mathcal{E}(P_t f) \leq e^{- \frac{\lambda_N}{C_N} t} \mathcal{E}(f) 
\end{equation} 
where $\lambda_N$ is the constant in the estimate \eqref{T_2>cT}  $$ \lambda_N =  \| b_N \|_{2}^{-1}  - 2(C_{pin}(N) + C_{int}(N) ) \|b_N\|_2 \| b_N^{-1}\|_2.$$  Thanks to the equivalence of $ \mathcal{T}$ and $ | \nabla_z|^2$ in \eqref{equiv T and grad}, we get the above convergence in the non-perturbed setting with equivalence-constant $  \operatorname{max}\left(1,\|b_N^{-1}\|_2 \right) \| b_N \|_{2}. $\\
In particular, both the Boltzmann entropy $H_{\mu}(P_tf \mu)$, given by \eqref{Boltzmann functional}, and the Fisher information $I_{\mu}(P_tf \mu)$, given by \eqref{Fisher info},  decay:

\begin{align} H_{\mu}(P_tf \mu ) + I_{\mu}(P_tf \mu) \le \frac{\|b_N\|_2}{\operatorname{min}\left(1,\|b_N^{-1}\|_2^{-1}\right)} e^{- \frac{\lambda_N}{C_N} t} \Big( H_{\mu}(f \mu ) + I_{\mu}(f \mu)   \Big)
\end{align}
%\begin{equation} I_{\mu}( P_t f \mu) = O\Big( \| b_N^{-1} \|_{2} \| b_N \|_{2} e^{ - \frac{\lambda_N}{C_N} t} \Big) \quad \text{and}\quad H_{\mu}(P_tf \mu ) = O\Big( e^{ - \frac{\lambda_N}{C_N} t} \Big)
%\end{equation}

 Estimates on $\|b_N\|_2$ regarding $N$ are given by the Proposition \ref{propos of induction} and so we conclude that \begin{align} H_{\mu}(P_tf \mu ) + I_{\mu}(P_tf \mu)  \le N^3 e^{-2\kappa_0 N^{-6} t} \Big( H_{\mu}(f \mu ) + I_{\mu}(f \mu)   \Big).
 \end{align}
\end{theorem}
Let us motivate in the next Remark the advantages of working in entropy when studying convergence in times for many particle systems. In short Relative Entropy and the Wasserstein metrics behave good with the dimension, whereas $L^2$- norm provides very bad, in $N$, estimates. 
\begin{remark}[Chaoticity and the Perks of Entropy/Wasserstein]
In general, in order to study convergence to equilibrium/NESS for many particle systems, it is natural to work with the so-called \textit{extensive} functionals which are functionals subadditive or even additive (proportional) in $N$. Typical examples of extensive functionals are the relative Entropy and the Wasserstein distances. In particlular, for the Entropy when we work with de-correlated data $f_t^{\otimes N}$ (i.e. chaotic data) we can write $$ H(f_t^{\otimes N})=NH(f_t) $$ and for the Wasserstein-2 distance $$   W_2(f_t^{\otimes N}, g_t^{\otimes N}) \le N^{1/2} W_2(f_t,g_t)$$ where for the second inequality we used that for $ k\ge 1$ and $a_i$ nonnegative constants, $(a_1+\cdots+a_n)^k \le n^{k-1}(a_1^k+ \cdots+a_n^k).$
So that if we have convergences in times  of the following form $$ H(f_t^{\otimes N}) \le C_N e^{-\lambda_N t}H(f_0^{\otimes N}) =  N C_N e^{-\lambda_N t} H(f_0)  \le C_1 N C_N e^{-\lambda_N t} $$ and $$ W_2(f_t^{\otimes N}, g_t^{\otimes N}) \le \tilde{C}_N e^{-\tilde{\lambda}_N t} W_2(f_0^{\otimes N}, g_0^{\otimes N}) \le C_2 N^{1/2} \tilde{C}_Ne^{-\tilde{\lambda}_N t}  $$
where $C_1,C_2$ do not depend on $N$. 
Whereas an $L_2$-estimate would give $$ \|f_0^{\otimes N} \|_2 \sim \| f_0\|_2^N.$$ So the convergence would take the form $$ \|f_t^{\otimes N}\|_2 \le C_N' e^{-\lambda_N' t} \|f_0^{\otimes N}\|_2 \leq \|f_0\|_2^NC_N' e^{-\lambda_N' t} .  $$
Then we need to wait until $ t \ge \frac{N \operatorname{ln}( \|f_0\|_2 )}{\lambda_N}$ and only then the convergence to equilibrium will start taking place. 
\end{remark}

From Theorem \ref{maintheorem} we get  an exponential rate of order bigger than $N^{-3}$. In the purely harmonic case, we have rate of order $1/N^{\gamma}$ where $ 1 \leq \gamma \leq 3.$\\

 We mention finally the following Proposition that gives directly the same  lower bound on the spectral gap (given the estimates on $\|b_N\|_2$ by Proposition \ref{propos of induction}): 
\begin{proposition}[Lower bound on the spectral gap of the harmonic chain] \label{Veselic} It is true that the spectral gap of the harmonic chain $\rho$ has a lower bound $$\rho \gtrsim \mathcal{O}(N^{-3} ) . $$ 
  \end{proposition}

\noindent
This lower bound is in fact the optimal rate in the case of the harmonic homogeneous chain.  In the work \cite{BM19} it is proven that $\rho = \mathcal{O}( N^{-3})$ by exploiting the form of the matrix $M$, \eqref{matrix M} and more specifically using  information on the spectrum of the discrete Laplacian. There, we study also the case of disordered chains. Unlike the homogeneous case here where the decay is \textit{polynomial}, in a disordered chain the spectral gap decays at an \textit{exponential rate} in terms of $N$. \\

\begin{remark}
Note that a  generalized version of $\Gamma$ calculus has been applied for a toy model of the dynamics \eqref{eq: SDE} by  P. Monmarch\'{e}, \cite{Mon15}: working with the unpinned, non-kinetic version, with convex interaction and given that the center of the mass is fixed, he proves the same kind of convergences and ends up with explicit and optimal $N$-dependent rates, of order $\mathcal{O}( N^{-2})$,  for the overdamped dynamics. 
\end{remark}

\subsection{Plan of the paper}  \hyperref[Section with Bakry-Emery theory]{Section 2}  contains an introduction to the
Bakry-Emery theory and an explanation of the method used. In \hyperref[section with functional ineqaualities]{Section 3} we obtain the estimates leading to the proof of Proposition \ref{log sobolev inqlt}. In \hyperref[Entropic convergence]{Section 4} and \hyperref[section with Wasserstein distance]{Section 5} we give the proof of Theorem \ref{converg in entropy} and Theorem \ref{maintheorem} respectively. Finally in \hyperref[section of induction]{Section 6} we prove Propositions \ref{propos of induction} and \ref{Veselic}.

\subsection{Notation} \label{Notation section} $\{e_i\}_{i=1}^n$  denote the elements of the canonical basis in $\mathbb{R}^n$ and  $| \cdot|$ to denote the Euclidean norm on $\mathbb{R}^n$, from the usual inner product $\langle \cdot, \cdot \rangle$. For a square matrix $A = (a_{ij})_{1 \leq i,j \leq n}  \in \mathbb{C}^{n \times n}$, we write  $\|A\|_2$ for the operator (spectral) norm, induced by the Euclidean norm for vectors :  $$ \| A \|_2 = \max_{x \in \mathbb{R}^n} \frac{ | Ax|_2}{|x|_2} = ( \text{maximum eigenvalue of } A^HA)^{1/2}  $$ and $A^*$ for the complex conjugate transpose $A^H=\bar{A}^T.$ 
We also write $A^{1/2}$ for the square root of a (positive definite) matrix $A$, \textit{i.e}. the matrix such that $A^{1/2} A^{1/2}=A$.   Moreover, by $C_b^{\infty}(\mathbb{R}^n)$ we denote the space of the smooth and bounded functions, by $\nabla_z$ we denote the gradient on $z$-variables in a metric space $X$ with respect to the Euclidean metric. We write $\mathcal{P}_2(\mathbb{R}^n)$ for the space of the probability measures on $\mathbb{R}^n$ that have second moment finite, \textit{i.e}. $$\mathcal{P}_2(\mathbb{R}^n) = \Big\{ \rho \in \mathcal{P}(\mathbb{R}^n): \int_{\mathbb{R}^n} |x|^2 d \rho(x) < \infty  \Big\}.$$
 $[N]$ denotes the set $\{1,2,\dots, N \}$.

%We use the above matrix notation only in the proof of Proposition \eqref{propos of induction}. 

\section{Carr\'{e} du Champ operators and curvature condition} \label{Section with Bakry-Emery theory}
Consider a Markov semigroup $P_t$ generated by an infinitesimal generator $\mathcal{L}: D(L) \subset L^2(\mu) \rightarrow L^2(\mu)$, where $\mu$ is the invariant measure of the dynamics. Here we restrict ourselves to the case of the diffusion operators and we associate a bilinear quadratic differential form $\Gamma$, the so-called \textit{Carr\'{e} du Champ} operator, and it is defined as follows: for every pair of functions $(f,g) $ in $ C^{\infty}\times C^{\infty}$ \begin{align}\label{Gamma operator} \Gamma(f,g):= \frac{1}{2} \Big( \mathcal{L}(fg) - f \mathcal{L}g-g \mathcal{L}f  \Big). \end{align}   
In other words $\Gamma$ measures the \textit{default of the distributivity} of $\mathcal{L}$.   \\                     
Then we can naturally define its iteration $\Gamma_2$, where instead of the multiplication we use the action of $\Gamma$:  \begin{equation} \Gamma_2(f,g):= \frac{1}{2} \Big( \mathcal{L}(\Gamma(f,g)) -\Gamma(f, \mathcal{L}g) - \Gamma(g, \mathcal{L}f)  \Big) \end{equation}
%Following Bakry-Emery's classical notation, we write $\Gamma(f):=\Gamma(f,f)$ and the same for $\Gamma_2$.\\

From the  theory of $\Gamma$-calculus we have that a curvature condition of the form \begin{equation} \label{classical Gamma_2 geq Gamma} \Gamma_2(f,f) \geq \lambda \Gamma(f,f) 
\end{equation} for all $f$ in a suitable algebra $\mathcal{A}$ dense in the $L^2(\mu)$ domain of $ \mathcal{L}$  and $\lambda >0$  is equivalent to the following gradient estimate $$ \Gamma \big( P_t f, P_t f \big) \leq e^{-2\lambda t} P_t (\Gamma (f,f)),\quad t \ge 0 $$ which implies a Log-Sobolev inequality (and thus Poincar\'{e} inequality), see \cite{BakEm83} or \cite[Section 3]{Bak04}.
\\

\noindent
Note that here the case is not the spatially homogeneous one and the generator of the dynamics described by \eqref{eq: SDE} is hypoelliptic, not elliptic, since the noise acts only on $2$ out of $2N$ variables of our phase space. So, the main problem in this case here is the lack of ellipticity, since the curvature condition \eqref{classical Gamma_2 geq Gamma} requires the ellipticity of the generator.  More specifically, for the operator \eqref{generator} we can not bound $\Gamma_2 $ by $ \Gamma$ from below since after explicit calculations we have the formulas \begin{align*} \Gamma(f,f) =2 \gamma_1 T_L (\partial_{p_1}f)^2 + 2  \gamma_N T_R (\partial_{p_N}f)^2\end{align*} while \begin{align*} \Gamma_2(f,f) = 2 (\gamma_1 T_L)^2 (\partial_{p_1}^2f)^2 + 2 (\gamma_N T_R)^2 (\partial_{p_N}^2f)^2 + 2 T_L \gamma_1^2 (\partial_{p_1}f) ( \partial_{q_1}f) + \\ + 2 T_R \gamma_N^2  (\partial_{p_N}f) ( \partial_{q_N}f) + \Gamma(f,f) .
\end{align*} Since we can not control the terms $ \partial_{p_i}f  \partial_{q_i}f, $ we can not bound $\Gamma_2$ from below by $\Gamma$. In cases like this, we say that the particle system has $ -\infty$  Bakry-Emery curvature. \\

\noindent
\subsection{Description of the method}
In order to overcome this problem,  we are doing the following:\\ 
 
\begin{enumerate} 
\item Firstly we perturb the classical $\Gamma$ theory, by defining a new quadratic form, different, but equivalent, to the $| \nabla_z f|^2$ that will play the role of the $\Gamma$ functional. This will spread the noise from $p_1$ and $p_N $ to all the other degrees of freedom as well. The general idea comes from Baudoin \cite{Bau17}.  We  make  a suitable choice of a positive definite matrix to define a new quadratic form that will replace the $\Gamma$ functional, so that we obtain a 'twisted' curvature condition: an estimate of the form \eqref{classical Gamma_2 geq Gamma}. This will imply also a perturbed gradient estimate, and thus a Poincar\'{e} and Log-Sobolev inequality. 

As Villani introduced in \cite{Villani09}, in order to deal with a hypocoercive situation in $H^1$- setting, one can perturb the norm to an equivalent norm, so that the desired convergence results can be deduced with this new norm. Then one can have convergence in the usual norm thanks to their equivalence. Here, instead of the norm, we perturb the gradient and thus the $\Gamma$ \textit{Carr\'{e} du Champ}, and work with a generalised $\Gamma$- theory. 

\item Second, in order to make our estimates quantitative we construct the matrix that we use to perturb the gradient. In  \hyperref[section of induction]{Section 6}, this matrix is constructed  as a solution of a sequence of \textit{continuous Lyapunov} equations and every step of the sequence corresponds to the spreading of noise and  dissipation to the next oscillator from both ends until the center of the chain. In fact this will be presented as follows: adding noise \textit{from the left to the right} until the $N$-th particle and \textit{from the right to the left} until the $1$-st particle, as shown in the Figure \ref{Diagram arrows}. These  steps will allow us to get estimates on how the rate of convergence to NESS behave with $N$. \\ 

\noindent
For those familiar with the method of H\"{o}rmander we describe briefly here the similarity with the spreading of dissipation-mechanism: in H\"{o}rmander's theory the  \textit{smoothing} mechanism is the one transferred through the interacting particles inductively by the use of commutators:  the generator has the form $$\mathcal{L}=X_0+X_1^2+X_N^2$$ where \begin{align*} X_0=p\cdot \nabla_q-\nabla_qH \cdot \nabla_p-\gamma p_1\partial_{p_1}- \gamma p_N \partial_{p_N}\quad \text{and}\quad X_i=\sqrt{T_i}\partial_{p_i} .\end{align*} Then $[\partial_{p_1},X_0]=-\partial_{p_1}+ \partial_{q_1} $. Now commuting $\partial_{q_1}$ with the first order terms of the generator: $[\partial_{q_1},X_0]= \partial_{q_1q_1}H \partial_{p_1}-\partial_{q_1q_2}H \partial_{p_2}$. Given that $\partial_{q_1q_2}H$ is non-vanishing we have 'spread the smoothing mechanism' to $p_2$. Continue like that, commuting the 'new' variable with the first order terms of $\mathcal{L}$,  inductively we will cover all the particles of the chain. \\  
\end{enumerate}
\begin{figure}
\[
\scalemath{1.2}{
\xymatrix{
p_1 \ar[d]  &p_2 \ar[d]
&\cdots& \cdots  & p_{N-1} \ar[d] & p_N \ar[d] \\
q_1 \ar[ru] & q_2  \ar[ru] & \cdots & \cdots &q_{N-1} \ar[lu] & q_N \ar[lu]}}
\]
\caption{Spreading of dissipation by commutators as in H\"{o}rmander's hypoellipticity theory. }\label{Diagram arrows}
\end{figure}
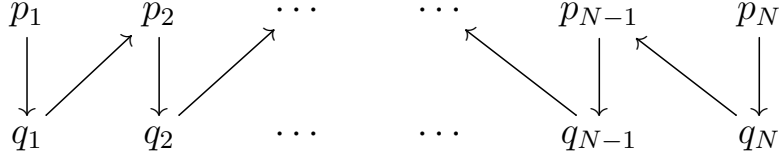
\medskip

\section{Functional inequalities in the perturbed setting } \label{section with functional ineqaualities}

\noindent
The goal is to apply a 'twisted' Bakry-Emery machinery, introduced by Baudoin in Section 2.6 of \cite{Bau17}. For that we define a perturbed quadratic form by using a positive definite matrix: this matrix is chosen as the solution of the following \textit{continuous Lyapunov equation} 

\begin{equation} \label{tupos gia b'} b_N M+M^Tb_N = \Pi_N .
\end{equation}
Here  $$M= \begin{bmatrix} \Gamma & -I \\ B & 0 \end{bmatrix}\  $$ and $\Pi_N$ is the  $2N \times 2N$ diagonal block matrix with strictly positive entries defined by $$ \Pi_N = \text{diag}(2 T_L, 1, \dots,1,2 T_R, 1, 1, \dots, 1,1). $$ 

The eigenvalues of $M$ have strictly positive real part (\cite[Lemma 5.1]{JPS17}) and the right hand side of \eqref{tupos gia b'} is positive definite. Therefore there exists a positive solution of \eqref{tupos gia b'}. This is a well known and classical result of Lyapunov  and it can be found for instance in \cite[page 224]{Gan59}: this equation first arose in connection with stability of linear dynamical systems. From A.M. Lyapunov's  monograph \cite[Section 20]{Lyap47}  follows a Theorem giving necessary and sufficient conditions for the eigenvalues of a real matrix to have negative real parts. \begin{theorem}[Lyapunov] Let an arbitrary negative definite quadratic form $W(z,z)$. There corresponds a positive definite quadratic form $V(z,z)$ such that if $$ \frac{d}{dt}z=-Mz\quad \text{then}\quad \frac{d}{dt} V(z,z)=W(z,z)$$ if and only if all the eigenvalues of $-M$ have negative real parts (that is $-M$ is stable) .
\end{theorem}
 If one chooses $W(z,z)=-z\ \Pi_N\ z^T$,  a matrix reformulation of the above Theorem gives necessary and sufficient condition for existence of positive solution of \eqref{tupos gia b'}. \\
%From the equation \eqref{tupos gia b'} we can see that $b$ should be symmetric as well.\\  

\noindent
Define the following quadratic quantity for $f,g \in C^{\infty}(\mathbb{R}^{2N})$, \begin{align}\label{perturbed Gamma operator} \mathcal{T}(f,g):=  \nabla_z f\  b_N \nabla_z g^T + \nabla_z g\  b_N \nabla_z f^T \end{align} so that $$  \mathcal{T}(f,f)=2 \nabla_zf\ b_N \nabla_zf^T.$$ Then the functional $$ \mathcal{T}_2(f,f)= \frac{1}{2} \Big( \mathcal{L} \mathcal{T}(f,f) - 2\mathcal{T}(f, \mathcal{L} f) \Big).$$

 \noindent
 Here $\mathcal{T}(f,f)$ is always positive since $b_N \geq 0$  in contrast with the original operator $\Gamma$, our perturbed quadratic form $\mathcal{T}$ is related to $\mathcal{L}$ only indirectly through the different steps of commutators. 
\\

We have an equivalence of the following form between $ \mathcal{T}$ and $| \nabla_z|^2$: \begin{equation} \label{equiv T and grad} \frac{1}{ \| b_N^{-1} \|_{2}} | \nabla_zf|^2 \leq \mathcal{T}(f,f) \leq  \| b_N \|_{2} | \nabla_z f|^2. \end{equation}

\begin{proposition} \label{Prop of T_2>T} With the above notation and for a fixed number of particles $N$, there exists constant $\lambda_N$ that depends on the spectral norm of the matrix $b_N$ and the bounds of the perturbing potentials $C_{pin}, C_{int}$ such that for $f \in C^{\infty}(\mathbb{R}^{2N})$, $$ \mathcal{T}_2(f,f) \geq \lambda_N \mathcal{T}(f,f). $$ 
\end{proposition}
\begin{proof}   
Using the form of the generator $\mathcal{L}$ as in \eqref{generator} : $$
\mathcal{L} = -zM \cdot \nabla_z -\nabla_q \Phi(q) \cdot \nabla_p + \gamma T_L \partial_{p_1}^2 + \gamma T_R \partial_{p_N}^2 $$ where $\Phi$ is the function that corresponds to the perturbing potentials,
 we write \begin{align*} 2 \mathcal{T}_2(f,f) &= \mathcal{L} \mathcal{T}(f,f)- 2 \mathcal{T}(f,\mathcal{L}f) = \mathcal{L} \mathcal{T}(f,f)- 2\nabla_zf\ b_N \nabla_z \mathcal{L}f^T - 2\nabla_z \mathcal{L} f\ b_N \nabla_z f^ T
 \end{align*}

From the  $(-zM \cdot \nabla_z)$- part of $\mathcal{L}$, the calculations are \begin{align*} -2 z M \nabla_z \nabla_z f\  &b_N \nabla_zf^T - 2\nabla_z f\ b_N  (z M \nabla_z \nabla_z f)^T\\ &+ 2\nabla_z f b_N (\nabla_z( z M \cdot \nabla_z f))^T+ 2\nabla_z( z M \cdot \nabla_z f) b_N \nabla_z f^T \\ &=  -2 z M \nabla_z \nabla_z f\  b_N \nabla_zf^T - 2\nabla_z f\ b_N  (z M \nabla_z \nabla_z f)^T\\ & +2 \nabla_zf\ b_N M\ \nabla_zf^T+2\nabla_z f\ b_N\ (zM \nabla_z \nabla_z f)^T \\ & + 2 \nabla_zf\ M^Tb_N\ \nabla_zf^T +2 zM\nabla_z \nabla_zf\ b_N\ \nabla_zf^T .\end{align*} 
The first term is cancelled with the last one and the second with the fourth one. \\
Similarly, from the $(-\nabla_q \Phi(q) \cdot \nabla_p)$- part of $\mathcal{L}$ we have  \begin{align*}-2 \nabla_q &\Phi \nabla_p \nabla_z f\ b_N \nabla_zf^T- 2\nabla_zf\ b_N (\nabla_q \Phi \nabla_p \nabla_z f)^T \\ &+ 2\nabla_zf\ b_N  ( \nabla_z f \nabla_z  \nabla_q \Phi )^T + 2\nabla_zf\ b_N (\nabla_q \Phi \nabla_z \nabla_p f)^T \\ &+ 2\nabla_q \Phi \nabla_p \nabla_z f\ b_N\ \nabla_zf^T +2\nabla_zf \nabla_z \nabla_q \Phi\ b_N\ \nabla_zf^T. \end{align*}
The first term is cancelled with the fifth one and the second with the fourth one.

Finally for the second order terms of the generator we write \begin{align*}
 4 \gamma T_L\ & \nabla_z \partial_{p_1} f\ b_N\  \nabla_z \partial_{p_1} f^T + 2\gamma T_L  \nabla_z \partial_{p_1}^2 f\ b_N\ \nabla_zf^T + 2\gamma T_L \nabla_zf\ b_N\  \nabla_z \partial_{p_1}^2 f^T \\&-2 \gamma T_L \nabla_z f\ b_N\  \nabla_z \partial_{p_1}^2 f^T- 2\gamma T_L \nabla_z \partial_{p_1}^2f\  b_N\  \nabla_zf^T \\ &+ 4 \gamma T_R\  \nabla_z \partial_{p_N} f\ b_N\ \nabla_z  \partial_{p_N} f^T + 2\gamma T_R  \nabla_z \partial_{p_N}^2f\ b_N\ \nabla_zf^T + 2\gamma T_R \nabla_zf\ b_N\ \partial_{p_N}^2 \nabla_zf^T \\ &- 2 \gamma T_R \nabla_z f\ b_N\  \nabla_z \partial_{p_N}^2f^T - 2\gamma T_R \nabla_z \partial_{p_N}^2f\ b_N\  \nabla_zf^T.
\end{align*}

\noindent
%For the first two parts of the generator we just apply the chain rule, while for the terms that are due to the noise (the second order terms of the generator) we just used that they commute with $\nabla_z$.
 We eventually end up with 
\begin{align*} 
\mathcal{T}_2&(f,f) = \nabla_zf\ b_N  M\ \nabla_zf^T + \nabla_zf\ M^Tb_N\  \nabla_zf^T + \nabla_zf\ b_N \text{Hess} (\Phi)^T\ \nabla_zf^T \\&+ \nabla_zf\ \text{Hess}(\Phi)  b_N\ \nabla_zf^T   + 2 \gamma T_L \mathcal{T}( \partial_{p_1} f,\partial_{p_1} f) + 2\gamma T_R \mathcal{T}( \partial_{p_N}f,\partial_{p_N}f) \\ & \geq  \nabla_zf\ (b_N  M + M^Tb_N)  \nabla_zf^T + \nabla_zf\ b_N \big( \text{Hess}(U_{pin}) + \text{Hess}(U_{int}) \big) \nabla_zf^T \\ &+ \nabla_zf\ \big( \text{Hess}(U_{pin})+ \text{Hess}(U_{int}) \big)^T b_N \nabla_zf^T \\ &=  \nabla_zf\ (b_N M+ M^T b_N) \nabla_zf^T + \nabla_zf\  ( b_N \text{Hess}(U_{pin})+ \text{Hess}(U_{pin})^T b_N ) \nabla_zf^T \\ &+ \nabla_zf\ \big( b_N \text{Hess}(U_{int})+ \text{Hess}(U_{\text{int}})^T b_N \big) \nabla_zf^T   
\end{align*}
where  for the second inequality we used that the terms $\mathcal{T}( \partial_{p_i} f,\partial_{p_i} f)$ for $i=1,N$, are positive. From the boundedness assumption on the operator norms of the Hessians for both perturbing potentials, and using that $b_N$  solves  the equation \eqref{tupos gia b'}, we get the following 
\begin{align*}
&\mathcal{T}_2(f,f)  \geq  \nabla_zf\ \Pi_N \nabla_zf^T - 2\| b_N\ \text{Hess}(U_{pin}) \|_2 | \nabla_z f|^2- 2\| b_N\ \text{Hess}(U_{int}) \|_2 | \nabla_zf|^2  \\ &\geq  \nabla_zf \Pi_N \nabla_z f^T - 2\| b_N \|_2 \sup_z \| \text{Hess}(U_{pin}(z)) \|_2 | \nabla_z f|^2 - \\ &-2\| b_N \|_2 \sup_z \|\text{Hess}(U_{int})(z) \|_2 | \nabla_zf|^2 \\ &\geq  | \nabla_z f |^2  -2 (C_{pin}(N) +C_{int}(N)) \| b_N\|_{2} |\nabla_zf|^2 \\ &\geq   \| b_N \|_2^{-1} \mathcal{T}(f,f) - 2 (C_{pin}(N) +C_{int}(N)) \|b_N \|_2 \| b_N^{-1}\|_2 \mathcal{T}(f,f).
\end{align*}

We conclude that \begin{equation} \label{T_2>cT}
\mathcal{T}_2(f,f) \geq \lambda_N \mathcal{T}(f,f),
\end{equation}
where \begin{equation} \label{lambda_N}  \lambda_N =  \| b_N \|_{2}^{-1}  - 2(C_{pin}(N) + C_{int}(N) ) \|b_N\|_2 \| b_N^{-1}\|_2
\end{equation} 
and we choose the quantity $(C_{pin}(N) + C_{int}(N) )$ in a way so that $\lambda_N$ is positive. In particular (after getting information on $\|b_N\|_2$ from Proposition \ref{propos of induction})  we require, see \eqref{behaviour of bounds} $$ C_{pin}(N) + C_{int}(N) = \mathcal{O} ( N^{-6}). $$ 
% $$ C_{pin}(N) + C_{int}(N)  < \frac{1}{ 2\| b_N \|_{2}^2 \|b_N^{-1}\|_2 }.  $$
\end{proof}

% $$ C_3 > \frac{(T+ \Delta T)}{2} \left( \frac{2 (\Delta T)}{ e^a-1}- \frac{2 (\Delta T) e^{-a(N-1)}}{e^a-1} + (T+ \Delta T) \Delta T \right) + \frac{\epsilon_{pin} + \epsilon_{int}}{4}  $$ and asymptotically, as $N$ becomes larger the above requirement takes the form $$  C_3 > \frac{(T+ \Delta T)}{2} \left( \frac{2 (\Delta T)}{ e^a-1} + (T+ \Delta T) \right) + \frac{\epsilon_{pin} + \epsilon_{int}}{4}.  $$

%\begin{remark}
%Note that we could prove the above Lemma by using techniques from \cite{AE14}, but then it would be much harder to estimate the $N$ dependence of the appearing constants, since that would acquire explicit knowledge of the eigenvectors of the matrix $\tilde{M}$. 
%\end{remark}
\noindent
We state now the following lemma that gives the 'twisted' gradient bound. 
\begin{lemma}[Gradient bound] \label{lemma of gradient bdd} If the operator $\mathcal{L}$ satisfies the curvature condition \eqref{T_2>cT} for some $\lambda_N$ and for $f \in C_c^{\infty}(\mathbb{R}^{2N})$, we have the following perturbed gradient estimate
\begin{align} \label{pert grad estimate}
 \mathcal{T}(P_tf,P_tf) \leq e^{-2\lambda_N t} P_t(\mathcal{T}(f,f)). \end{align}
\end{lemma}
\begin{proof}
If $\mathcal{T}$ is compactly supported we consider the functional $$ \Psi(s)= P_s \big( \mathcal{T}(P_{t-s}f,P_{t-s}f) \big),\ s \in [0,t]$$ for $f \in C_c^{\infty}(\mathbb{R}^{2N})$ and for fixed $t$. Since from the semigroup property we have $$\frac{d}{ds} P_s= \mathcal{L}P_s= P_s \mathcal{L},$$ by differentiating and using the above inequality we get $$ \frac{d}{ds} \Psi(s)=2P_s \big( \mathcal{T}_2(P_{t-s}f, P_{t-s}f) \big) \geq 2\lambda _N P_s \big( \mathcal{T}(P_{t-s}f,P_{t-s}f) \big) = 2\lambda_N \Psi(s)$$  and since $\Psi(0)=\mathcal{T}(P_tf,P_tf),\ \Psi(t)=P_t(\mathcal{T}(f,f))$, by Gr\"{o}nwall's lemma we get the desired inequality for every smooth and bounded function $f$.\\

In general we need $\mathcal{T}(P_t f,P_tf)$ to belong in $L^{\infty}(\mathbb{R}^{2N})$ because then we know that $P_s\big(\mathcal{T}(P_{t-s}f,P_{t-s}f)\big)$ is well defined. So we do the following:\\

First we take $W(p,q)= 1+ |p|^2 + |q|^2$ as a Lyapunov structure that satisfies the following conditions: $W >1$, $ \mathcal{L} W \leq C W$, the sets $ \{ W \leq m \}$ are compact for each $m$, and $\mathcal{T}(W) \leq C W^2$. This $W$  satisfy the  conditions thanks to the bounded-Hessians assumption, \textit{i.e}. $ | \nabla(U_{int}+ U_{pin})| $ will be Lipschitz. In particular, for the inequality $ \mathcal{L} W \leq C W $ using Cauchy-Schwarz and Young's inequalities, we write \begin{align*}
\mathcal{L} W &= 2 p \cdot q -2q B \cdot p-2 p \cdot \nabla_q \Phi - 2 \gamma_1 p_1^2 - 2 \gamma_N p_N^2 + 2T_L \gamma_1 +2 T_R \gamma_N \\ &\leq 2 |p||q|+ 2|Bq||p| + 2 | \nabla_q \Phi | |p| +  2T_L \gamma_1 +2 T_R \gamma_N \\ &\leq |p|^2 + |q|^2 + C_{ C_{lip}, \|B\|_2} (|p|^2 + |q|^2)+  2T_L \gamma_1 +2 T_R \gamma_N \\ & \leq \max \big\{ \max(1, C_{  C_{lip},\| B \|_{2}}),  2T_L \gamma_1 +2 T_R \gamma_N \big\} (1+|p|^2+|q|^2)= C_1 W
\end{align*}  
while the inequality $ \mathcal{T}(W) \leq C_2 W^2$ obviously holds. So we end up with the same constant by taking as $C=: \max\{C_1,C_2\}.$
\\

Now using the function $W$ combined with a localization argument as in the work by F.Y. Wang  \cite[Lemma 2.1]{Wan12} or \cite[Theorem 2.2]{Bau16} we prove the boundedness of $ \mathcal{T} (P_tf,P_tf).$
 Consider $h \in C_c^{\infty}([0,\infty))$ such that $h\vert_{[0,1]}=1$ and $ h\vert_{[2,\infty)}=0$ and define $$\phi_n = h(W/n)\quad \text{and}\ \mathcal{L}_n=\phi_n^2 \mathcal{L}.$$ Then $\mathcal{L}_n$ is compactly supported in $B_n:= \{ W \leq 2n \} $. Let $P_t^n$ be the semigroup generated by $\mathcal{L}_n$, which is given as the bounded solution of $$ \mathcal{L}_n P_t^nf =\partial_t P_t^n f\quad \text{for}\ f \in L^{\infty}(\mathbb{R}^{2N}). $$ Then we also have that          \begin{align*} P_t^n &\stackrel{n\to \infty}{\to}
 P_t\quad \text{pointwise}. \end{align*}
 We do the 'interpolation semigroup argument' as before for $\mathcal{L}_n$ and for $f \in C_c^{\infty}(\mathbb{R}^{2N})$ supported in $\{W\leq n\}$. Define $$ \Psi_n(s) = P_s^n(\mathcal{T}(P_{t-s}^nf,P_{t-s}^nf)),\quad s \in [0,t] $$ for fixed  $t>0$,  $n \ge 1$ applied to a fixed point in the support.\\
  It is true, due to the properties of $W$,  that $\mathcal{T}(P_t^n f,P_t^n f) \leq C_{f,t} $ with $C_{f,t}$ independent of $n$ and so we have a bound on $\mathcal{T}(P_t^nf,P_t^nf )$ uniformly on the set $\{W\leq n\}$. Indeed \begin{align*} \Psi'_n(s)&= P_s^n( \mathcal{L}_n \mathcal{T}(P_{t-s}^nf,P_{t-s}^nf) - 2 \mathcal{T} (\mathcal{L}_nP_{t-s}^nf,P_{t-s}^nf)) \\ &= P_s^n ( 2h_n^2 \mathcal{T}_2(P_{t-s}^nf,P_{t-s}^nf)-4h_n\mathcal{L} P_{t-s}^nf \mathcal{T}(h_n,P_{t-s}^nf))   \\ & \ge  P_s^n ( 2h_n^2 \lambda_N \mathcal{T}(P_{t-s}^nf,P_{t-s}^nf)-4h_n\mathcal{L} P_{t-s}^nf \mathcal{T}(h_n,P_{t-s}^nf))  \\ &\ge P_s^n ( 2h_n^2 \lambda_N \mathcal{T}(P_{t-s}^nf,P_{t-s}^nf)-4P_{t-s}^n\mathcal{L}_n f \mathcal{T}(\log h_n,P_{t-s}^nf))  \\ &\ge P_s^n \big( 2h_n^2 \lambda_N \mathcal{T}(P_{t-s}^nf,P_{t-s}^nf)-4 \| \mathcal{L} f \|_{\infty} \sqrt{\mathcal{T}(\log h_n,\log h_n)} \sqrt{\mathcal{T}(P_{t-s}^nf,P_{t-s}^nf)} \big) \\ & \stackrel{\text{Young's ineq.}}{\geq} P_s^n \big( -(2| \lambda_N | + 2) \mathcal{T}(P_{t-s}^nf,P_{t-s}^nf) -C_1 \mathcal{T}(\log h_n,\log h_n) \big) 
  \end{align*} 
  with $C_1$  constant independent of $n$. About the last term: $$ \mathcal{T}(\log h_n, \log h_n)  = -\frac{1}{n^2h_n^2} h'(W/n)^2 \mathcal{T}(W) \le \frac{C}{h_n^2} $$ with $C$ independent of $n$. Now calculate $$ L_n \left( \frac{1}{h_n^2} \right) = -\frac{2h'(W/n)\mathcal{L}W}{nh_n} - \frac{2h''(W/n)\Gamma(W)}{n^2h_n} + \frac{6h'(W/n)^2 \Gamma(W)}{n^2h_n^2} \leq \frac{C_2}{h_n^2} $$ with $C_2>0$ some constant again independent of $n$ (from assumptions on the Lyapunov functional $W$). Therefore $$   P_s^n \left( \frac{1}{h_n^2} \right) \leq \frac{e^{sc_2}}{h_n^2}. $$ Combining this last estimate with the above bounds we end up with $$ \Psi_n'(s) \ge  -(2| \lambda_N | + 2) \Psi_n(s)- C_3$$ and $C_3=C_3(f,t) $ is again independent on $n$. We integrate in time from $0$ to $t$ and  we get the desired boundedness on $\{ W \leq n\}$. \\

Now if $d'$ is the intrinsic distance induced by $\mathcal{T}$ $$d'(x,y)=\sup_{\mathcal{T}(f,f) \le 1} |f(x)-f(y)|, $$ from the above bound we have that $$ |P_t^nf(x) - P_t^nf(y)| \le C d'(x,y)$$ for $n $ large enough so that  $x,y \in \{W \le n\}$ and $f \in C_c^{\infty}(\mathbb{R}^{2N})$ with support in  $\{W \le n\}$. $C$ does not depend on $n$ (from before), so passing to the limit we have 
  $$ |P_t f(x) - P_t f(y)| \le C d'(x,y)$$ and so $\mathcal{T}(P_tf, P_tf) $ is also bounded. Now we can repeat the standard Bakry-Emery calculations as in the beginning of the proof. 
\end{proof}
%If we were working on some compact space, the above would be naturally explained for $f$ smooth, as $P_t$ would immediately preserve smoothness (from the hypoellipticity of the generator) and we would have also the desired boundedness. Here we want to justify this for a non-compact space. 

We refer also  to the discussion in the book \cite[Section 3.2.3, page 145]{BGL14} for more details about for which classes of functions does this gradient bound hold. In \cite[Theorem 3.2.4]{BGL14} they give an argument in order to justify the above gradient bound for $f \in L^2(\mu)$-domain of the generator $L$ of the diffusion proccess. They assume though reversibility of their reference measure $\mu$. 

\begin{remark} \label{Rem: L2 gradient estimate} Note that using the equivalence of $ \mathcal{T}$ and $| \nabla_z|^2$: $$ \frac{1}{  \| b_N^{-1} \|_{2}} | \nabla_zf|^2 \leq \mathcal{T}(f,f) \leq  \| b_N \|_{2} | \nabla_z f|^2,$$  we get the following $L^2$- gradient estimate \begin{align} \label{L2 gradient estimate}
 | \nabla_zP_tf|^2 &\leq \|b_N \|_{2} \| b_N^{-1} \|_{2}\ e^{-2\lambda_N t} P_t \big( |\nabla_zf|^2 \big)
\end{align}
\end{remark}

Once we have a curvature condition of the form \eqref{T_2>cT} we are also able to show that the stationary measure satisfies a Poincar\'{e} inequality. 
\begin{proposition} \label{Poincare inql}Let $\mathcal{L} $ be the generator of the dynamics described by the SDEs \eqref{eq: SDE}. Let $\Gamma$ the operator defined in \eqref{Gamma operator}, while $\mathcal{T}$ the perturbed quadratic form defined in \eqref{perturbed Gamma operator}. If $f \in C^{\infty}(\mathbb{R}^{2N})$ and a constant $\lambda_N >0$ exists so that the inequality $ \mathcal{T}_2(f,f) \geq \lambda_N \mathcal{T}(f,f)$ holds, the unique invariant measure $\mu=f_{\infty} $ from the Theorem \ref{maintheorem} satisfies a Poincar\'{e} inequality $$ \text{Var}_{\mu} (f) \leq C_N \int_{\mathbb{R}^{2N}} \mathcal{T}(f,f) d\mu.$$
where $C_N =  \frac{ T_L \| b_N^{-1} \|_{2}}{\lambda_N}$  and the non-perturbed Poincar\'{e} inequality takes the form $$ \text{Var}_{\mu} (f) \leq  \| b_N \|_{2}\ C_N \int_{\mathbb{R}^{2N}} | \nabla_z f |^2 d\mu.$$
\end{proposition}
\begin{proof} For $f \in  C^{\infty}(\mathbb{R}^{2N})$, we consider the functional $$\Psi(s) = P_s( (P_{t-s}f)^2),\ s \in [0,t].$$ By differentiating we have \begin{align*}
\Psi'(s) &= \mathcal{L} P_s ((P_{t-s} f)^2)-2 P_s (P_{t-s}f \mathcal{L} P_{t-s}f) \\ &=2P_s \big( \Gamma(P_{t-s}f, P_{t-s}f) \big).
\end{align*}
Now by integrating from $0$ to $t$ \begin{align*}
P_t(f^2)- &(P_tf)^2 = 2 \int_0^t P_s ( \Gamma(P_{t-s}f,P_{t-s}f)) ds \leq 2 T_L \int_0^t P_s ( | \nabla P_{t-s} f |^2  ) ds \\ &\leq 2 T_L \| b_N^{-1} \|_{2} \int_0^t P_s ( \mathcal{T}(P_{t-s}f,P_{t-s}f)) ds \\ &\leq 2 T_L \| b_N^{-1} \|_{2} \int_0^t P_s( e^{-2\lambda_N(t-s)} P_{t-s} \mathcal{T}(f,f) ) ds \\ &= 2 T_L \| b_N^{-1} \|_{2}\ e^{-2 \lambda_N t} P_t \mathcal{T}(f,f) \int_0^t e^{2 \lambda_N s} ds \\ &= 2 T_L  \| b_N^{-1} \|_{2}\ e^{-2 \lambda_N t} P_t \mathcal{T}(f,f) \left( \frac{e^{2 \lambda_N t}-1}{2 \lambda_N} \right) \\ &= T_L  \|  b_N^{-1} \|_{2}\ \frac{1-e^{-2\lambda_N t}}{\lambda_N} P_t \mathcal{T}(f,f)
\end{align*}
where in the first inequality we used that $$\Gamma(f, f) = T_L (\partial_{p_1}f)^2 + T_R (\partial_{p_N}f)^2 \leq \max\{T_L,T_R\}  | \nabla f |^2,$$ for the second we used the gradient bound from Lemma \ref{lemma of gradient bdd} and just right after that, the semigroup property.\\
Now letting $t$ to $\infty$, thanks to the ergodicity we have the desired inequality with constant $C_N= \frac{T_L \| b_N^{-1} \|_{2}}{\lambda_N}$.
\end{proof}

\noindent
So the constant in the (non-perturbed) Poincar\'{e} inequality depends on $N$ through the norms $\|b_N\|_2, \|b_N^{-1}\|_2$, since 
$$  \tilde{C}_N =  \frac{ T_L \| b_N\|_{2} \| b_N^{-1} \|_{2} }{2 \| b_N \|_2^{-1} - 2\| b_N \|_{2} \| b_N^{-1}\|_2 (C_{pin}(N)+C_{int}(N))} .$$ 

\noindent
In fact it is possible to show a stronger pointwise gradient bound, that will be exploited in the proof of a Log-Sobolev inequality for the invariant measure of the dynamics. 
 
\begin{proposition}[Strong gradient bound] \label{Prop of strong gradient bound} 
It is true that for $f \in C_c^{\infty}$, $\forall t\geq 0$, \begin{equation} \label{strong gradient bound} \mathcal{T}(P_{t} f,P_{t} f) \leq  \Big( P_t ( \sqrt{\mathcal{T}(f,f)}) \Big)^2 e^{-2 \lambda_N t}. 
\end{equation}
\end{proposition}

\begin{flushleft}
\emph{Note.} This is a better estimate than  \eqref{pert grad estimate} in  Lemma \ref{lemma of gradient bdd} because of Jensen's inequality.
\end{flushleft}
\begin{proof} The rigorous justification, \textit{i.e.} boundedness of  $ \sqrt{\mathcal{T}(P_{t-s}f,P_{t-s}f)}$),  of the following formal calculations is exactly like in the proof of Lemma \ref{lemma of gradient bdd}.\\
 Here for $f \in C_c^{\infty} $, and for fixed $t \geq 0$, instead we define $$\Phi(s) = P_s \Big( \sqrt{\mathcal{T}(P_{t-s}f,P_{t-s}f)}\Big),\ s \in [0,t].$$ By differentiating and performing the standard calculations we have
 \begingroup
 \begin{footnotesize} 
  \begin{align*} 
&\Phi'(s) = P_s \Big(  \mathcal{L}  (\sqrt{\mathcal{T}(P_{t-s}f,P_{t-s}f)}) - \frac{\nabla \mathcal{L}P_{t-s}f^T b_N \nabla P_{t-s}f + \nabla P_{t-s}f^T b_N \nabla \mathcal{L} P_{t-s}f  }{2 \sqrt{\mathcal{T}(P_{t-s}f,P_{t-s}f)}} \Big) \\ =& P_s \Big( \mathcal{L}  (\sqrt{\mathcal{T}(P_{t-s}f,P_{t-s}f)}) + \frac{2 \mathcal{T}_2(P_{t-s}f,P_{t-s}f)-\mathcal{L} \mathcal{T}(P_{t-s}f,P_{t-s}f)}{2 \sqrt{\mathcal{T}(P_{t-s}f,P_{t-s}f)}} \Big) \\ =& P_s \Bigg( \frac{1}{\sqrt{\mathcal{T}(P_{t-s}f,P_{t-s}f)}} \Big( -\Gamma\big(\sqrt{\mathcal{T}(P_{t-s}f,P_{t-s}f)},\sqrt{\mathcal{T}(P_{t-s}f,P_{t-s}f)}\big) \\&+ 2  \mathcal{T}_2(P_{t-s}f,P_{t-s}f) \Big) \Bigg) \\ =&  P_s \Bigg( \frac{1}{\sqrt{\mathcal{T}(P_{t-s}f,P_{t-s}f)}} \Big( 2 \mathcal{T}_2(P_{t-s}f,P_{t-s}f)\\ & -\frac{T_L(\mathcal{T}(\partial_{p_1}P_{t-s}f,\partial_{p_1}P_{t-s}f))^2 + T_R(\mathcal{T}(\partial_{p_N}P_{t-s}f,\partial_{p_N}P_{t-s}f))^2 }{4 \mathcal{T}(P_{t-s}f,P_{t-s}f)}   \Big) \Bigg) \\ \geq& P_s \Bigg( \frac{1}{4 \mathcal{T}(P_{t-s}f,P_{t-s}f)^{3/2}} \bigg( 8 \lambda_N (\mathcal{T}(P_{t-s}f,P_{t-s}f))^2 + 
8T_L\big( \mathcal{T}(\partial_{p_1} P_{t-s}f)  \big)^2 \\&+ 8T_R\big( \mathcal{T}(\partial_{p_N} P_{t-s}f)  \big)^2
  - T_L \big(\partial_{p_1} \mathcal{T}(P_{t-s}f,P_{t-s}f) \big)^2 - T_R  \big( \partial_{p_N} \mathcal{T}(P_{t-s}f,P_{t-s}f)\big)^2 \bigg) \Bigg) \\ &\geq P_s \left( \frac{8\lambda_N (\mathcal{T}(P_{t-s}f,P_{t-s}f))^2}{4 \mathcal{T}(P_{t-s}f,P_{t-s}f)^{3/2}} \right) \\ &= 2 \lambda_N \Phi(s) 
\end{align*}
\end{footnotesize}
\endgroup
where in the first inequality we have used the formula $$\mathcal{T}_2(f,f) \geq \lambda_N \mathcal{T}(f,f)+ T_L \mathcal{T}(\partial_{p_1}f,\partial_{p_1}f) + T_R \mathcal{T}(\partial_{p_N}f,\partial_{p_N}f)$$ that we obtained in Proposition \ref{Prop of T_2>T}, that $$\Gamma(f,g)= T_L (\partial_{p_1}f) (\partial_{p_1} g) + T_R (\partial_{p_N}f) (\partial_{p_N} g)$$ for the generator of the dynamics \eqref{eq: SDE} and that $\mathcal{T}$ and $\partial_{p_1}$ obviously commute. 
Now from Gr\"{o}nwall's lemma we get \begin{align*}
\Phi(s) \geq e^{2 \lambda_N t} \Phi(0)\ \Rightarrow\  \mathcal{T} ( P_{t} f, P_tf) \leq e^{-4 \lambda_N t} \Big( P_t ( \sqrt{\mathcal{T}(f,f)}) \Big)^2.  
\end{align*}
\end{proof} 

Now this pointwise, strong gradient bound implies a Log-Sobolev inequality.

\begin{proof}[Proof of Proposition \ref{log sobolev inqlt}]
 For $f \in C_c^{\infty}(\mathbb{R}^{2N}) $, we introduce the functional $$H(s) = P_s \Big( P_{t-s}f \log P_{t-s}f \Big)$$ and following again Bakry's recipes, we get \begin{align*}
 H'(s) &= P_s \Big( \mathcal{L}(P_{t-s}f \log P_{t-s}f) - \mathcal{L}P_{t-s}f \log P_{t-s}f - \mathcal{L}(P_{t-s}f) \Big) \\ &= P_s \Big( \Gamma(P_{t-s}f,\log P_{t-s}f )\Big) \\ &= P_s \left( \frac{T_L (\partial_{p_1}P_{t-s}f)^2}{P_{t-s}f} + \frac{T_R (\partial_{p_N}P_{t-s}f)^2}{P_{t-s}f}  \right) \\ &= P_s \left(\frac{\Gamma(P_{t-s}f,P_{t-s}f)}{P_{t-s}f} \right) \\ &\leq T_L \| b_N^{-1} \|_{2} P_s \left(\frac{\mathcal{T}(P_{t-s}f,P_{t-s}f)}{P_{t-s}f} \right) \\ &\leq T_L  \| b_N^{-1} \|_{2} P_s \left( e^{-2\lambda_N (t-s)} \frac{(P_{t-s} (\sqrt{\mathcal{T}(f,f)}))^2}{P_{t-s}f } \right) \\ &\leq T_L \| b_N^{-1} \|_{2} P_t  \left( \frac{\mathcal{T}(f,f)}{f} \right)e^{-2\lambda_N (t-s)}
\end{align*}
where for the second inequality we used the bound from Proposition \ref{Prop of strong gradient bound}, while for the last inequality we used  Jensen's and the Markov property of the semigroup. 
Now integrating from $0$ to $t$, we get \begin{align*}
 H(t)-H(0) &\leq \frac{ T_L \| b_N^{-1} \|_{2} }{2\lambda_N} (1- e^{-2\lambda_N t}) P_t \left( \frac{\mathcal{T}(f,f)}{f} \right) \\ & \leq \frac{  T_L \| b_N^{-1} \|_{2} \| b_N \|_{2}}{2\lambda_N} (1- e^{-2\lambda_N t}) P_t \left( \frac{| \nabla_z f |^2}{f} \right)
\end{align*}
Letting $t \to \infty $ and thanks to the ergodicity of the semigroup, we get the LSI with constant $\tilde{C}_N= \frac{  T_L \| b_N^{-1} \|_{2} \| b_N \|_{2} }{2\lambda_N } $ corresponding  to the constant  with the non-perturbed Fischer information.
\end{proof}

\noindent
We note here that since a Log-Sobolev inequality holds,  we also have  (see \cite{Gross93}) that our semigroup is \emph{hypercontractive}. Hypercontractivity says that at time $t>0$ the semigroup is regularizing from $L^p$ to $L^{q(t)}$ for $$q(t)-1=\exp(4t/\tilde{C}_N)(p-1).$$ For the notion of hypercontractivity we refer for instance to \cite[Theorem 2.5]{Bak04} and references therein. Hypercontractivity and the validity of a Log-Sobolev inequality, is in fact stronger result than $L^2$-exponential convergence since it implies  stronger exponential convergence in entropy which implies $L^2$ by a linearization, for that see \cite[Proposition 2.3]{Wang17}. For similar results  in hypoelliptic cases, for instance the kinetic Fokker-Planck equation, see \cite{Bau17}, where a Log-Sobolev inequality is proven by  similar local Bakry-Emery computations.

\section{Entropic Convergence to equilibrium} \label{Entropic convergence}

If $\mu$ is the invariant measure of the system, we will prove here convergence to NESS in Entropy  as stated in Theorem \ref{converg in entropy}:  with respect to the functional $$\mathcal{E}(f) := \int_{\mathbb{R}^{2N}} f \log f+ f \mathcal{T}( \log f, \log f) d\mu.$$ 

\begin{proof}[Proof of Theorem \ref{converg in entropy}]
We consider the functional $$ \Lambda(s) = P_s \Big(  P_{t-s}f \mathcal{T}( \log P_{t-s}f, \log P_{t-s}f) \Big) $$ and by differentiating and repeating similarly the steps from the Propositions \ref{Prop of strong gradient bound} and \ref{log sobolev inqlt} we end up with \begin{align*}
\Lambda'(s) &= P_s \mathcal{L} \Big(  P_{t-s}f \mathcal{T}( \log P_{t-s}f, \log P_{t-s}f) \Big) - 2 P_s \left( P_{t-s}f  \mathcal{T}\left(\log P_{t-s}f  , \frac{\mathcal{L} P_{t-s}f}{P_{t-s}f} \right) \right) \\ -& P_s\Big( \mathcal{L} P_{t-s}f  \mathcal{T}( \log P_{t-s}f, \log P_{t-s}f)  \Big)\\ & = P_s \Big(P_{t-s}f \mathcal{L} \mathcal{T}(\log P_{t-s} f)  \Big) + 2 P_s \Big(  \Gamma( P_{t-s}f, \mathcal{T}(\log P_{t-s}f )) \Big) \\ & -2 P_s \Big( P_{t-s}f \mathcal{T} \big(\log P_{t-s}f, \Gamma(\log P_{t-s}f)+ \mathcal{L}(\log P_{t-s}f)\big)  \Big) \\ &= 2 P_s \Big(  P_{t-s}f \mathcal{T}_2(\log P_{t-s}f) \Big) \\ &\geq 2\lambda_N    P_s \Big(  P_{t-s}f \mathcal{T}(\log P_{t-s}f) \Big)
\end{align*}
 where  we have used that for the second equality \begin{align*} &\mathcal{L} ( \log P_{t-s}f ) = \frac{\mathcal{L} P_{t-s} f }{P_{t-s}f} - \Gamma( \log P_{t-s} f )\ \text{and}\ \\& \mathcal{T} \big(\log P_{t-s}f, \Gamma(\log P_{t-s}f, \log P_{t-s}f) \big) =  \Gamma \big(\log P_{t-s}f, \mathcal{T}(\log P_{t-s}f, \log P_{t-s}f) \big)  \end{align*} and  in the last inequality we used the bound \eqref{T_2>cT}. Now, integrating against the invariant measure $\mu$, and applying the Log-Sobolev inequality from Proposition \ref{log sobolev inqlt}, we write \begin{align*}
 \int_{\mathbb{R}^{2N}} \Lambda'(s) d\mu &\geq \frac{\lambda_N}{C_N} \int_{\mathbb{R}^{2N}} P_s \Big( P_{t-s}f \log P_{t-s}f \Big) d\mu \\ &+ \lambda_N \int_{\mathbb{R}^{2N}} P_s \Big( \mathcal{T}(\log P_{t-s}f,\log P_{t-s}f) P_{t-s}f \Big) d\mu \\ &\geq \min \left( \frac{\lambda_N}{C_N},\lambda_N \right) \int_{\mathbb{R}^{2N}} \Lambda(s) d\mu  = \frac{\lambda_N}{C_N} \int_{\mathbb{R}^{2N}} \Lambda(s) d\mu.
 \end{align*}
 Finally, from Gr\"{o}nwall's inequality we have $$ \int \Lambda(0) d\mu \leq e^{- \frac{\lambda_N}{C_N} t} \int \Lambda(t) d\mu $$ or equivalently the desired convergence, thanks to the invariance of the measure. \\
 Since $C_N=\frac{ T_L\| b_N^{-1} \|_{2}}{2\lambda_N},$ we write $$  \frac{\lambda_N}{C_N} = \frac{2\lambda_N^2}{T_L \|b_N^{-1}\|_2} \gtrsim 2\kappa_0 N^{-6}.$$ The last inequality is due to the estimates later in \eqref{size of lambda_N}. 
\end{proof}

\section{Convergence to equilibrium in Kantorovich-Wasserstein distance } \label{section with Wasserstein distance}

We recall the definition of the Kantorovich-Rubinstein-Wasserstein $L^2$-distance $W_2(\mu, \nu)$ between two  probability measures $\mu, \nu$, for some metric $d$: $$ W_2(\mu,\nu)^2 = \inf \int_{\mathbb{R}^N \times \mathbb{R}^N} d(x,y)^2  d\pi(x,y) $$
where the infimum is taken over the set of all the couplings, \textit{i.e}. the joint measures $\pi$ on $ \mathbb{R}^N \times \mathbb{R}^N$ with left and right marginals $\mu$ and $\nu$ respectively.
\\

It is easy to see that $W_2$ satisfies the definition of a metric, whenever of course $d$ is a metric, and so it is indeed a metric. We also restrict ourselves on the subspace $\mathcal{P}_2(\mathbb{R}^{2N})$, where $\mu$ and $\nu$ have second moments finite, so that their distance $W_2(\mu,\nu)$ will be finite. For more information on this distance see for instance \cite{Vill09} and references therein. Also, even though that convergence in Monge-Kantorovich-Wasserstein distance is a weak convergence comparatively with convergence in entropic sense or convergence in $L^2$ for instance, it can be defined on the more natural subspace $\mathcal{P}_2(\mathbb{R}^{2N})$ of probability measures. 
%One can deduce after that similar contraction results in classical Sobolev spaces, since  $W_2$ distances are controlled from the so called Fourier based distances $$ d_2(f,g) := \sup_{k \in \mathbb{R}^{2N}-\{0\}}\frac{| \hat{f}(k) -\hat{g}(k) |}{k^2}$$ (for this see Proposition 2.12 in \cite{CarTos07}) and then control Sobolev norms by controlling $d_2$ distances together with arbitrarily larger Sobolev norms, see theorem 4.1 in \cite{CGT99}. 
Finally, in contrast with entropy, it is more general to assume finite Wasserstein distance. 
\\

For a measurable function $f$ and $x \in \mathbb{R}^{2N}$, we define the upper-gradient with respect to the metric $d$ $$ | \nabla_d f| (x) := \lim_{r \to 0}\ \sup_{0< d(x,y) \le r} \frac{| f(x)-f(y)|}{d(x,y)}.$$% Note that $\sup_{x \in X} |\nabla_d f(x)| < \infty $ if and only if $f$ is Lipschitz continuous. \\

What we use here to get convergence in Kantorovich-Wasserstein distance is that the gradient estimate of type \eqref{L2 gradient estimate} is equivalent to an estimate in Wasserstein distance (Kuwada's duality, \cite{Kuw09}). More specifically, we have the following Theorem, here  stated only in the Euclidean space with the Lebesgue measure $ (\mathbb{R}^{2N}, |\cdot|, \lambda)$:
\begin{theorem}[Theorem 2.2 of \cite{Kuw09}] \label{theorem Kuwada}  Let the Markov semigroup $P$  on $\mathbb{R}^{2N}$, that has a  continuous density with respect to the Lebesgue measure. For $c>0$, the following are equivalent: \begin{itemize}
\item[(i)] For all probability measures $\mu, \nu$  we have, $$ W_p (P_t^* \mu, P_t^* \nu)  \leq c\ W_p (\mu,\nu)$$ where $W_p$ denotes the Wasserstein distance associated with the Euclidean distance. 
\item[(ii)] When $p>1$ and $q$ its H\"{o}lder conjugate,  for all bounded and Lipschitz functions $f$ and $ z \in \mathbb{R}^{2N}$, $$ |\nabla P_t f | (z) \leq c\ P_t \big( | \nabla f|^q \big)(z)^{1/q} $$ where this estimate is associated with the Lipschitz norm defined just above.
\end{itemize} 
\end{theorem}
Now, considering all the above, we  prove Theorem \ref{maintheorem} from which we get both uniqueness of an invariant measure and convergence to NESS of the semigroup $P_t$.

\begin{proof}[Proof of Theorem \ref{maintheorem}]
The convergence follows if we apply Kuwada's duality from Theorem \ref{theorem Kuwada} since we have the estimate \eqref{L2 gradient estimate} with  $c= \| b_N^{-1} \|_{2}^{1/2} \| b_N \|_{2}^{1/2}.$
\end{proof}

\begin{remark}
 A convergence to equilibrium in total variation norm for a similar small perturbation of the harmonic oscillator chain, has been shown recently in \cite{RAQ18}. There, a version of Harris' ergodic Theorem was applied making it possible to treat more general cases of the oscillator chain with different kind of noises, as well. However, this is a non-quantitative version of Harris' Theorem, which provides no information on the dependency of the convergence rate in $N$.  
\end{remark}

\section{Estimates on the spectral  norm of $b_N$ } \label{section of induction}
 First, let us state the following Proposition on the optimal exponential rate of convergence for the \textit{ purely harmonic chain}.
 
\begin{proposition}[Proposition 7.1 and 7.2 (3) in \cite{BM19}]\label{optimal sg for harmonic}  
 The optimal order of the spectral gap of the dynamics which evolution is described by the generator \eqref{generator}, \textit{without the perturbing potentials}, is given by the order of  $$ \rho=  \inf\{\text{Re}(\mu) : \mu \in \sigma(M) \}.$$
 The spectral gap approaches $0$ as $N$ goes to infinity, and the rate should be at least  of order $\mathcal{O}(1/N)$. 
 \end{proposition}
 
\begin{proof} We exploit the results by Arnold and Erb in \cite{AE14} or  by Monmarch\'{e} in \cite[Proposition 13]{Mon15}  saying that working with an operator of the form $$L f(x) = -(B x) \cdot \nabla_x f (x) + \text{div}(D \nabla_x f )(x)$$ under the conditions that (i) no non-trivial subspace of $\text{Ker} D$ is invariant under $B^T$ and (ii) the matrix $B$ is positively stable, \textit{i.e.} all the eigenvalues have real part greater than $0$, then the associated semigroup has a unique invariant measure and  if $\rho = \inf\{\text{Re}(\mu) : \mu \in \sigma(B) \}>0$, 
the sharp exponential rate of the above generator (of an Ornstein-Uhlenbeck process) is at least $\rho - \epsilon$ and at most $\rho$. This holds for every $\epsilon \in (0, \rho)$, concluding  the first statement of the Proposition. In particular, when $m$ is the maximal dimension of the Jordan block of $C$ corresponding to the eigenvalue $\lambda$ such that $\text{Re}(\lambda) = \rho$,  the quantity $(1+t^{2(m-1)})e^{-2\rho t}$ is the optimal one regarding the long time behaviour, \cite{Mon15}. This implies that the spectral gap of the generator is $\rho-\epsilon$, whereas the constant in front of the exponential $c(\epsilon,m):= \sup_{t}(1+t^{2(m-1)})e^{-2\epsilon t}.$  \\

When we look  at the purely  harmonic chain this is our case as well:  the first condition is equivalent to the hypoellipticity of the operator $L$, \cite[Section 1]{Ho69}, and our generator \eqref{generator} is indeed hypoelliptic: it is proven, \cite[Section 3, page 667]{EPR99b} and  \cite[Section 3]{Car07}, for more general classes of potentials than the quadratic ones, that the generator satisfies the rank condition of H\"{o}rmander's hypoellipticity Theorem, \cite[Theorem 22.2.1]{Hormander}.  
Also the matrix $M$ is stable for every $N$, \text{i.e.} $\text{Re}(\lambda) \ge 0$  for all the eigenvalues $\lambda$, see \cite[Lemma 5.1]{JPS17}.\\

For the second conclusion of the Proposition:\\
Since the matrix $$M= \begin{bmatrix} \Gamma & -I \\ B & 0 \end{bmatrix} $$ we have $$ \text{Tr}(\Gamma) = \text{Re}(\text{Tr}(\Gamma))= \text{Re}(\text{Tr}(M)) = \sum_{\lambda \in \sigma(M) }\text{Re}(\lambda). $$ Now since the $\text{Tr}(\Gamma)$ does not depend on the number of oscillators, the sum of $2N$ (counting multiplicity) positive terms\footnote{Note again that the $\inf \{\text{Re}(\lambda) \} $ is strictly positive, see \cite[Lemma 5.1(2)]{JPS17}} in the r.h.s. should also be uniformly bounded in $N$. Since  $$  \sum_{\lambda \in \sigma(M) }\text{Re}(\lambda) \geq 2N\inf \{\text{Re}(\lambda): \lambda \in \sigma(M)\} $$ we have that $$2N\inf \{\text{Re}(\lambda): \lambda \in \sigma(M)\}\  \text{is bounded asymptotically with $N$}$$ which implies that \begin{equation} 
  \inf \{\text{Re}(\lambda): \lambda \in \sigma(M)\} \leq \frac{C}{2N} = \mathcal{O} \left( \frac{1}{N} \right) \end{equation} for some constant $C$ independent of $N$. Thus the smallest real part of the eigenvalues should have order less than $\frac{1}{N}$. 
\end{proof}

\begin{remark} \label{remark before the induction proof}
Let us remark  here that $B$ can be seen as the Schr\"odinger operator : $B=-c\ \Delta^N + \sum_{i=1}^N \alpha \delta_i$ where $c>0$, $  \Delta^N $ is the Dirichlet Laplacian on $ l^2( \{1,\dots,N \})$ and $\delta_i$ the projection on the $i$-th coordinate.  We give the following definition for the (discrete)  Laplacian on $ l^2( \{1,\dots,N \})$ with Dirichlet boundary conditions: $$-\Delta^N := \sum_{i=1}^{N-1} L^{i,i+1} $$ where $L^{i,i+1}$ are uniquely determined by the quadratic form \begin{align*} \langle u, L^{i,i+1} u \rangle &= (u(i)-u(i+1))^2\quad \text{with} \\ u(0)&=u(N+1)=0\qquad \text{Dirichlet b.c}. \end{align*} We will use this information just after the proof of Proposition \eqref{propos of induction}, to quantify the equivalence constants and the LSI constant from \hyperref[section with functional ineqaualities]{Section 3}. 
\end{remark}

\subsection{Strategy}
 We consider the sequence of symmetric  matrices $ \{ b_i \}_{i=0}^{N}$ such that for each $i$, $b_i =\begin{bmatrix} x_i & z_i \\ z_i^T & y_i \end{bmatrix}$  solves the following set of Lyapunov equations 
  
\begin{align*}
b_0 M+M^T b_0 &= \text{diag} \left( 2T_L,0,\dots,2T_R,0,\dots,0  \right)  := 2 \tilde{\Theta}\\ b_1 M + M^Tb_1 &=  \text{diag}\left(2 T_L,0, \dots, 0,2T_R,\frac{1}{2},0, \dots, 0,\frac{1}{2} \right):= \Pi_1= \begin{bmatrix} J_1^{(\Delta T)} &0 \\ 0& J_1^{(0)}  \end{bmatrix}
\\ b_2M+M^Tb_2&= \text{diag}\left(2T_L,\frac{1}{2},0,\dots,0,\frac{1}{2},2T_R,\frac{1}{2},\frac{1}{2},0,\dots,0,\frac{1}{2},\frac{1}{2}\right) := \Pi_2= \begin{bmatrix} J_2^{(\Delta T)} &0 \\ 0& J_2^{(0)}  \end{bmatrix}\\
&\vdots \\
b_{N} M+ M^T &b_{N}= \text{diag}(2T_L, 1, \dots,1,2T_R, 1, 1, \dots, 1,1) := \Pi_{N}
\end{align*}
\textit{i.e.} in every step we add half a unit in the diagonal of each diagonal block \textit{from both ends} until we have a full diagonal matrix in the r.h.s of the Lyapunov equation (with only nonzero entries). This steps represent the\textit{ spread of dissipation  from the endpoints of the chain to each particle} like the commutators would do in a classical hypoelliptic setting.  We could stop the process at the $\lfloor N/2 \rfloor +1 $ steps, but to ease the notation we do it $N$ times, and it will be like running two inductions simultaneously: one adding $1/2$ 's from the top left until the bottom right side and one from the bottom right to the top left side of the matrix. 
See Figure \ref{Diagram arrows} in the introduction for a visualization. 
\\

\noindent
There is a positive semi-definite, symmetric solution $b_i$ to each one of these equations $$ b_i = \int_0^{\infty} e^{-t M}  \Pi_i  e^{-tM^T}  dt. $$ The proof can be found in \cite[page 82]{SL61} for positive right hand side of the \textit{Lyapunov} equations. The same proof holds for nonnegative.\\

\subsection{Matrix equations on Lyapunov equation}
\begin{lemma} \label{lemm: blocks- equations} For $0 \le m \le N$, we have the following equations for the blocks $x_m,y_m$ and $z_m$ of the matrix $b_m$:  \begin{align}
  -z_m &=z_m^T + J_m^{(0)} \label{eq: 1} \\
  x_m &= B y_m +  \Gamma z_m  \label{eq: 2}\\
  -Bz_m+z_mB -& B J_m^{(0)}  = J_m^{(\Delta T)}- x_m \Gamma - \Gamma x_m  \label{eq: 3} \\
  y_m B - By_m &= \Gamma + z_m \Gamma + \Gamma z_m   \label{eq: 4a}.
  \end{align}
 Here $J_m^{(0)} =  \operatorname{diag}(1,1,\dots,1,0,\dots,0, 1,1,\dots,1)$ where the $0$'s start at $(m+1,m+1)$-entry and stop at $(N-(m+1),N-(m+1))$-entry, and\\ $J_m^{(\Delta T)} = \operatorname{diag}(2T_L,1,\dots,1,0,\dots,0,1,\dots,1,2T_R) $ where the $0$'s start at $(m+2,m+2)$-entry and stop at $(N-(m+2),N-(m+2))$-entry.
\end{lemma}

\begin{proof}
We consider $m$ s.t. $0 \le m \leq N$, where $b_m$ solves \begin{equation} \label{tupos gia b_m} b_m M +M^T b_m=  \Pi_m \end{equation}  and where $$\Pi_m = \begin{bmatrix} J_m^{(\Delta T)} &0 \\ 0& J_m^{(0)}  \end{bmatrix}. $$  

From \eqref{tupos gia b_m} and considering that $x_m$ and $y_m$ are symmetric matrices,  we get  $$  \begin{bmatrix} x_m\Gamma + \Gamma x_m +z_mB+Bz_m^T & -x_m+\Gamma z_m + B y_m \\ -x_m+ z_m^T \Gamma +y_m B & -z_m^T -z_m  \end{bmatrix} = \begin{bmatrix} J_m^{(\Delta T)} &0 \\ 0& J_m^{(0)}  \end{bmatrix}. $$ From that we get  \eqref{eq: 1} and \eqref{eq: 2} directly, and also that:
  \begin{align}
  Bz_m^T+z_mB & = J_m^{(\Delta T)}- x_m \Gamma - \Gamma x_m \label{eq: 3pr}
  \end{align}
 and by applying \eqref{eq: 1}  to \eqref{eq: 3pr} we get \eqref{eq: 3}.
 
 Also, using that $x_m$ and $y_m$ are required to be symmetric matrices,  from the transposed version of \eqref{eq: 2}, we get  the equation  $$ x_m = y_m B -z_m \Gamma - J_m^{(0)} \Gamma $$ which, combined with \eqref{eq: 2}, gives \begin{align}
 &\eqref{eq: 4a}\quad  \text{for}\ m\geq 1\ \text{and}   \\
 y_m B - By_m &=  z_m \Gamma + \Gamma z_m \quad \text{for}\ m=0  \label{eq: 4b}
 \end{align}
 \end{proof}

 We perform all the calculations from now on when the dimension of the block matrices, $N$,  is odd. The same calculations with minor differences hold when $N$ is even as well.\\

\subsection{Calculations for $m=0,1,2$}  Before we start analyzing the form of the block $z_N$, we first present in this subsection how each unit in the right hand side of the Lyapunov equation  \eqref{tupos gia b_m} for $0 \le m \leq N$ (that corresponds to the spread of noise on the system), affects the $z_m$ block of the solution $b_m$. \\ This, motivated by the strategy-subsection above,  is only to make it easier for the reader to follow on how perturbing the r.h.s. of the Lyapunov equation affects the solution in each sequential step. Then we analyze \hyperlink{mylink}{the final step} (proof of Lemma \ref{lemma: z_N block}) the result that we are interested in. Thus, the reader who is interested only in the proofs, and  not in the motivation behind them, might  skip this subsection.\\
 %We look at how each unit in the right hand side of the Lyapunov equation \eqref{tupos gia b_m} for $0 \le m \leq N$ (that corresponds to the spread of noise on the system), is affecting the $z_m$ block of the solution $b_m$.\\

For $m=0$: It has been computed in \cite{RLL67}, where they found exactly the elements of $z_0 := (z_{ij}^{(0)})_{1\leq i,j \leq N}$ when $a=0,c=1$, to be \begin{align}  \label{Lebowitz explicit solution} z_{1,j}^{(0)}=\frac{\operatorname{sinh}((N-j)\alpha)}{\operatorname{sinh}(N\alpha)} \end{align} for $\alpha$ constant such that $ \operatorname{cosh} (\alpha) = 1+\frac{1}{2\gamma}$. (It was done in the same manner with \cite[Section 11]{WU45} but there the case was $\Delta T =0$).  Here we describe briefly the steps:
 %\footnote{The calculations in \cite{RLL67} correspond to the step $m=0$, where they found exactly the elements of $z_0 := (z_{ij}^{(0)})_{1\leq i,j \leq N}$ when $a=0,c=1$. Their chain is modelled with rigidly fixed edges meaning that they consider achain beginning with an oscillator labelled $0$ and ending with one labelled $N+1$ under the hypothesis that $q_0=q_{N+1}=0$. Also consider that the first and the last particle are pinned with harmonic forces (modelling their attachment to a wall). Then they work with the corresponding matrix $B$ (\eqref{matrix B}) being the discrete Laplacian with Dirichlet boundary conditions. }
 first we notice that $z_0$ is antisymmetric since in \eqref{eq: 1} $J_0^{(0)}=0$, and second, by \eqref{eq: 3} we get that it has a \textit{Toeplitz}-form  \begin{align} \label{form of z_0} \scalemath{0.89}{ z_0 = 
 \begin{bmatrix} 
 0& z_{1,2}^{(0)} & z_{1,3}^{(0)}& z_{1,4}^{(0)} & \cdots & z_{1,N-1}^{(0)} & z_{1,N}^{(0)} \\
 -z_{1,2}^{(0)} &0 & z_{1,2}^{(0)} & z_{1,3}^{(0)} & \cdots & z_{1,N-2}^{(0)} & z_{1, N-1}^{(0)} \\
 -z_{1,3}^{(0)} & -z_{1,2}^{(0)} &0& z_{1,2}^{(0)} & \cdots & \quad & \cdots\\
 \quad & \quad & \ddots & \quad & \quad \\
 -z_{1,N-1}^{(0)} & -z_{1,N-2}^{(0)} & \quad & \quad & \quad & 0 & z_{1,2}^{(0)} \\
 -z_{1,N}^{(0)} & -z_{1,N-1}^{(0)} &\quad &  \quad & \quad & -z_{1,2}^{(0)} & 0   
\end{bmatrix} } :
\end{align}
 Indeed  note that the r.h.s of \eqref{eq: 3}  forms a bordered matrix  \[  \mleft[
\begin{array}{c|ccc|c}   \ast & \ast & \cdots & \ast & \ast \\ 
\hline
\ast &0 & \quad & 0 & \ast \\
\vdots & \quad& \ddots & \quad & \vdots \\
\ast & 0 &\quad  & 0 & \ast \\
\hline
\ast & \ast & \cdots & \ast & \ast
\end{array}
\mright] \] \textit{i.e.} only the bordered elements are non zero and so the l.h.s of \eqref{eq: 3} should also have this bordered form. Due to the tridiagonal form of $B$ we get a \textit{Toeplitz} matrix: in particular using that $B= -c \Delta^N + a I$, the l.h.s of \eqref{eq: 3} is \begin{align} \label{eq: 0 induction} z_0 ( -c \Delta^N + a I)-(-c \Delta^N + a I) z_0 = c(\Delta^N z_0 - z_0 \Delta^N ) =\mleft[
\begin{array}{c|ccc|c}   \ast & \ast & \cdots & \ast & \ast \\ 
\hline
\ast &0 & \quad & 0 & \ast \\
\vdots & \quad& \ddots & \quad & \vdots \\
\ast & 0 &\quad  & 0 & \ast \\
\hline
\ast & \ast & \cdots & \ast & \ast
\end{array}
\mright]  \end{align} and  equating the nonboundary-entries, due to the symmetry of $\Delta^N$ and the antisymmetry of $z_0$, we have that the elements of $z_0$ will be constant along the diagonals: indeed, for $1<i<N$, for the diagonal's entries of the equation \eqref{eq: 0 induction} we have \begin{align*} -cz_{i-1,i}^{(0)}-cz_{i+1,i}^{(0)}&+2cz_{i,i}^{(0)} -2c z_{i,i}^{(0)}+c z_{i,i-1}^{(0)} + c z_{i,i+1}^{(0)}=0 \\ \text{or}\quad 2cz_{i,i+1}^{(0)}&-2c z_{i-1,i}^{(0)} =0\\ \text{and so}\quad  z_{i,i+1}^{(0)}&=z_{i-1,i}^{(0)}. \end{align*} For the superdiagonal's entries of the equation \eqref{eq: 0 induction} \begin{align*}-cz_{i-1,i+1}^{(0)}+2cz_{i,i+1}^{(0)}-c&z_{i+1,i+1}^{(0)}+cz_{ii}^{(0)}-2cz_{i,i+1}^{(0)}+cz_{i,i+2}^{(0)}=0 \\ \text{or}\quad -cz_{i-1,i+1}^{(0)} &+cz_{i,i+2}^{(0)}=0 \\ \text{so}\quad z_{i-1,i+1}^{(0)}&=z_{i,i+2}^{(0)}.   \end{align*} We repeat these calculations through all the non-boundary entries of the matrix, and using the information we get from each one calculation, we end up with the Toeplitz form of $z_0$ in \eqref{form of z_0}.  \\

Now find that a solution to \eqref{eq: 4b} is a symmetric Hankel matrix which is antisymmetric about the cross diagonal and such that $(y_{1,j}^{(0)})_{j=1}^{N-1}= z_{1,j+1}^{(0)}.$ Then apply \eqref{eq: 2} to  get a formula for the entries of $x_0$ and from the  information due to  \eqref{eq: 3} about the bordered entries of $x_0$, end up with the linear equation $$K_0 \cdot \underline{z_0} =e_1.$$ Here $\underline{z_0},\ e_1 \in \mathbb{C}^{N-1}$ are the vectors  $\underline{z_0}=(z_{1,1}^{(0)}, \dots, z_{1,N-1}^{(0)})^T$,  $e_1=(1,0,\dots,0)^T$ and $K_0$ is a $(N-1) \times (N-1)$ symmetric Jacobi matrix whose entries depend on the (dimensionless) friction constant $\gamma$ and interaction constant $c$: $$ K_0= cB+ \gamma^{-1} I.$$ \\ Now solve the above equation using for example Cramer's rule and find an explicit formula for the $z_{1,j}^{(0)}$'s: the recurrence formula of the determinant of $K_0$ is the same formula of the  Chebyshev polynomials of the second kind, so using properties of these polynomials and  imposing appropriate initial conditions we end up with the form \eqref{Lebowitz explicit solution}.  
The detailed implementation of this is done in \cite{RLL67}. \\

 For $m\geq 1$ we use again the equation  \eqref{eq: 3}. In the first step we get that:\\ For $m=1$, \textit{i.e.} for the form of the $z_1$-block in $b_1$, the  elements $z_{1,1}^{(1)}, z_{N,N}^{(1)}$ in the main diagonal are $-1/2$.  The difference with the $m=0$ step is that $z_1$ is not antisymmetric anymore, since $1/2$ is added in the first entry of the diagonal (due to the form of $J_1^{(0)}$). So from \eqref{eq: 1} we write $$-z_{i,i}^{(1)}=z_{i,i}^{(1)}+1\quad \text{or}\quad z_{i,i}^{(1)}=-1/2\quad \text{for}\ i=1,N.$$ But we still have the bordered form in the r.h.s. of \eqref{eq: 3}, so we still have a \textit{Toeplitz}-form for $z_1$. \\ 
 
 In the next Lemma we give the form of the $z_2$ block of $b_2$. 
 
 \begin{lemma} [For $m=2$, form of $z_2$] \label{lemma z_2} For  the $z_2$-block of $b_2$ : $ z_2 = z_2^a - J_{2}^{(0)}$ and \begin{align*} \left\{ \begin{array}{ c l } z_{1,1}^{(2)}=z_{2,2}^{(2)}=z_{N,N}^{(2)}=z_{N-1,N-1}^{(2)}=-1/2\ \text{and}\ z_{i,i}^{(2)}=0\quad \text{otherwise} \\
 z_{1,2}^{(2)}+ z_{N,N-1}^{(2)} = 2\frac{ 1 +a+2c}{4c},\quad z_{N,N-2}^{(2)}+z_{1,3}^{(2)}=1 \\ 
 z_{N-k,N}^{(2)} =z_{1,k+1}^{(2)}\quad \text{for}\ 3\le k \le N-3 .\end{array} \right.\end{align*}
 The last property is that the \textit{Toeplitz} form is not perturbed in more than $2$ diagonals away from the centre. \\
% \begin{align*}
 %\left\{ \begin{array}{ c l } z_{ij}^{(2)} \pm \frac{1+a+2c}{2c}, &\text{for}\ i=j \pm 1, \\ 
   %z_{ij}^{(2)} \pm 1/2, &\text{for}\ i=j \pm 2, \\
   %z_{11}^{(2)}=z_{22}^{(2)}=z_{NN}^{(2)}=z_{N-1,N-1}^{(2)}=-1/2\ \text{and}\ z_{ij}^{(2)}=0\ &\text{for}\ 2<i=j \leq N.
  % \end{array} \right.
   %\end{align*}
 So we denote by $\mu_{a,c}:= \frac{1+a+2c}{4c}$ and we write: \begin{align*} \scalemath{0.89}{z_2 =  \begin {bmatrix} -\frac{1}{2} & z_{1,2}^{(2)} &z_{1,3}^{(2)} & z_{1,4}^{(2)}&\cdots &  \cdots  & z_{1,N-1}^{(2)} & z_{1,N}^{(2)} \\
 -z_{1,2}^{(2)} &-\frac{1}{2} &z_{1,2}^{(2)} -\mu_{a,c} & z_{1,3}^{(2)}+\frac{1}{2} & z_{1,4}^{(2)}  &\cdots & z_{1,N-2}^{(2)} & z_{1,N-1}^{(2)} \\
 -z_{1,3}^{(2)} & -z_{1,2}^{(2)}+\mu_{a,c} &0& z_{1,2}^{(2)}-\mu_{a,c} &\cdots &\cdots  &\quad &z_{1,N-2}^{(2)}\\
 %0& -\frac{1}{2} & \mu_{a,c} & 0 &- \mu_{a,c} &\cdots & \quad & 0\\
 \vdots & \quad &\quad & \ddots &\quad & \quad & \quad \\
  \vdots & \quad &\quad & \quad &\ddots & \quad & \quad \\
  \vdots & \quad &\quad & \quad&\quad &0& -z_{N,N-1}^{(2)}+\mu_{a,c} & -z_{N,N-2}^{(2)} \\
 z_{N,2}^{(2)} & z_{N,3}^{(2)} & \dots&\quad &\quad  & z_{N,N-1}^{(2)}-\mu_{a,c} &-\frac{1}{2}& -z_{N,N-1}^{(2)} \\
 z_{N,1}^{(2)}& z_{N,2}^{(2)} &\dots&\quad &\quad & z_{N,N-2}^{(2)} &z_{N,N-1}^{(2)} &-\frac{1}{2} 
  \end{bmatrix} }.
 \end{align*} 
 \end{lemma}
 \begin{proof}[Proof of Lemma \ref{lemma z_2}] 
 $z_2 $ is not antisymmetric but from \eqref{eq: 1} we immediately have that $ z_2=z_2^a-J_{2}^{(0)}$, where $z_2^a$ is antisymmetric.  So we work with $z_2^a$ and due to the antisymmetry we look only at the upper diagonal part of the matrix. \\
 
Here, besides that $z_2$ is not antisymmetric, the r.h.s  of \eqref{eq: 3} is not a bordered matrix anymore and also the matrix $BJ_2^{(0)}$ affects non boundary entries as well, in particular it adds the $(3 \times 2)$ top-left and bottom-right submatrices of $B$ to the $(3 \times 2) $ respective submatrices of $z_2$:\begin{align} \label{eq: 2 induction} \scalemath{0.92}{ c(\Delta^N z_2 - z_2 \Delta^N ) + (c\Delta^N -aI) \text{diag}\left(\frac{1}{2},\frac{1}{2},0,\dots,0,\frac{1}{2},\frac{1}{2}\right) = \mleft[
\begin{array}{c|c|ccc|c|c}   \ast & \ast &\ast & \cdots &\ast& \ast & \ast \\ 
\hline
\ast&1/2&0 & \quad & 0 &0& \ast \\
\hline
\ast& 0 &0 &\quad &0&0 &\ast \\
\vdots & \quad& \quad & \ddots &\quad&\quad& \vdots \\
\ast& 0& 0 & \quad&0&0&\ast \\
\hline
\ast & 0 &0 &\quad  &0& 1/2& \ast \\
\hline
\ast & \ast & \ast & \cdots &\ast&  \ast & \ast
\end{array}
\mright] .}\end{align}
%On the one hand,  we have the constants $\mu_{a,c}$ and $1/2$ added in the diagonals as described above, because of the $1/2$'s in the right side-matrix and the 2nd term in the above equation \eqref{eq: 2 induction}. On the other hand, 
%These changes  remain exactly the same along the diagonals (kind of a Toeplitz form again) and it does not affect more than $2$ diagonals away from the centre: 
Equating the entries that correspond to the zero-submatrix as drawn above we will have the same calculations as in the step $m=0$.\\
From \eqref{eq: 1} we have $z_{1,1}^{(2)}=z_{2,2}^{(2)}=z_{N,N}^{(2)}=z_{N-1,N-1}^{(2)}=-1/2$ and $z_{i,i}^{(2)}=0$ for $N-1>i>2$.\\  Looking at the $(2,2)$-entry and the $(2,3)$-entry of the equation \eqref{eq: 2 induction} we have respectively
 \begin{align*}
 -c z_{2,1}^{(2)} + 2c z_{2,2}^{(2)}-c z_{2,3}^{(2)}& +c z_{1,2}^{(2)} -2c z_{2,2}^{(2)}+c z_{3,2}^{(2)} -\frac{(a+2c)}{2} = \frac{1}{2}  \\ -c z_{2,2}^{(2)}+2c z_{2,3}^{(2)}-c z_{2,4}^{(2)}+&c z_{1,3}^{(2)}-2c z_{2,3}^{(2)}+c z_{3,3}^{(2)}=0
\end{align*}
and since $z_{i,j}^{(2)}=-z_{j,i}^{(2)}$ for $j \neq i$ from \eqref{eq: 1}, and also $z_{2,2}^{(2)}=-1/2, z_{3,3}^{(2)}=0$, we get  $$ z_{2,3}^{(2)}=z_{1,2}^{(2)}- \mu_{a,c}\quad \text{and}\quad z_{2,4}^{(2)}=z_{1,3}^{(2)}+1/2.$$
Now looking at the entries $(i,i)$ for $3 \le i\le N-2$ of equation \eqref{eq: 2 induction}, we write (as in the $0$-step): $$-cz_{i,i-1}^{(2)} + 2c z_{i,i}^{(2)} -cz_{i,i+1} + cz_{i-1,i}^{(2)} -2c z_{i,i}^{(2)} + cz_{i+1,i}^{(2)} =0 $$ which gives $$ z_{i-1,i}^{(2)}=z_{i,i+1}^{(2)},\quad 3 \le i\le N-2.$$  In particular \begin{align*} z_{i,i+1}^{(2)} &= z_{1,2}^{(2)}-\mu_{a,c}=- z_{N,N-1}^{(2)} + \mu_{a,c}\quad \text{and}\\ z_{i,i+2}^{(2)} &= z_{1,3}^{(2)}+ \frac{1}{2} = -z_{N,N-2}^{(2)} -\frac{1}{2} \end{align*} 
where the second equalities in both lines  are proved by looking at the reversed direction (bottom-right to top-left side of the matrix). Also for $k \ge 2$ and $ 1 \le i \le N-k$, look at $(i,i+k)$ entry of the equation \eqref{eq: 2 induction} and get $$ z_{i,i+k+1}^{(2)}=z_{i-1,i+k}^{(2)}.$$ This corresponds to the Toeplitz property that holds for all the diagonals apart from the $5$ central ones. Remember that for $m=0$ we end up with a Toeplitz matrix. 
%and repeating these calculations, by looking at the zero-submatrix of the r.h.s. of \eqref{eq: 2 induction} (like the $0$-step) we get $$ \text{for}\ i \ge 3,\quad z_{i,i+2}^{(2)}=z_{i-1,i+1}^{(2)}\quad \text{and so on.}$$ 
%The same calculations hold for the reversed induction, \textit{i.e.} from bottom-right to top-left side of the matrix. (In the $N$-th step later there are more details in the 'reversed calculations'). 
\end{proof}

\textit{In the $m$-th step of the sequence of these matrix equations},  for the $z_m$- block of $b_m$, the central $(4m-3)$ diagonals have a perturbed Toeplitz form: the elements across these diagonals on each line are changed by constants that depend on the coefficients $a,c$. 
 %since the elements across these diagonals increase by constants that depend on the coefficients $a,c$ and their absolute value is increasing as $m$ grows. 
 The resulting matrix $z_m$ is described in the following way, where $\mu_{a,c}:= \frac{1+a+2c}{4c}$:\\  
\begin{align*}
\left\{ \begin{array}{ c l }  z_{1,j}^{(m)}+z_{N,N-(j-1)}^{(m)} = m\mu_{a,c},\quad &\text{for j even},\  j\le m  \\
z_{1,j}^{(m)}+z_{N,N-(j-1)}^{(m)} = -m,\quad &\text{for j odd},\  j\le m \\
z_{N-j,N}^{(m)} = z_{1,j+1}^{(m)},\quad &\text{for}\ m <j<N-2,\qquad \textit{(Toeplitz form)}\\
z_{i,i}^{(m)} =-1/2, &\text{for}\ 1 \le m\ \text{and}\  i \ge N-m\\ z_{i,i}^{(m)}=0,\ &\text{for}\ m < i < N-m.
\end{array} \right.
\end{align*}
The explanation is the same as in the step $m=2$ but this holds for an arbitrary $m\leq N$.\\

\subsection{Preliminaries: compute the blocks $z_N$, $y_N, x_N$ of $b_N$} \label{subsection: compute z_N}

\begin{lemma}[Form of $z_N$ block] \label{lemma: z_N block} The matrix $z_N:= (z_{i,j}^{(N)})_{1 \le i,j \le N}$ is a real $N \times N$ matrix of the form $$ z_N = z_N^{a}- \frac{1}{2} I $$ where $z_N^{a}= [z_{i,j}^{(N),a}]$ is antisymmetric. We denote by $\mu_{a,c} := \frac{1+a+2c}{2c}$. $z_N$  has  the following perturbed Toeplitz form: for $2 \le i\le N-k$ and $1 \le k  \le N-2$,  \begin{align}\label{eq: pert Toeplitz form}  \left\{ \begin{array}{ c l }
 z_{i,i+k}^{(N),a} - z_{i-1,i+k-1}^{(N),a}=-\mu_{a,c},\ &\text{for}\ k\ \text{odd}  \\ 
  z_{i,i+k}^{(N),a} - z_{i-1,i+k-1}^{(N),a}=1,\ &\text{for}\ k\ \text{even} \end{array} \right.
   \end{align} and for the second and second-to-last line respectively: 
   \begin{align}   \left\{ \begin{array}{ c l }
 z_{2,k}^{(N),a} - z_{1,k-1}^{(N),a}=-\mu_{a,c},\quad z_{N-1,k}^{(N),a}-z_{N,k+1}^{(N),a}= - \mu_{a,c},\ &\text{for}\ k\ \text{odd}  \\ 
  z_{2,k}^{(N),a} - z_{1,k-1}^{(N),a}=1,\quad  z_{N-1,k}^{(N),a}-z_{N,k+1}^{(N),a}=1   &\text{for}\ k\ \text{even} \\
 \end{array} \right.
\end{align}
 Regarding the 'cross-diagonal' we have, for $1 \le k  \le N-2$, \begin{align} \label{x_ij cross diagonal}
   \left\{ \begin{array}{ c l }
 z_{i,i+k}^{(N),a} - z_{N-k-(i-1),N-(i-1)}^{(N),a}=(N-k-2i+1)\mu_{a,c},\ &\text{for}\ k\ \text{odd},\ 1\le i \le \frac{N-k}{2}  \\ 
  z_{i,i+k}^{(N),a} - z_{N-k-(i-1),N-(i-1)}^{(N),a}=k-N+2i-1,\ &\text{for}\ k\ \text{even},\ 1 \le i \le \frac{N-(k+1)}{2}. \end{array} \right.
   \end{align}
   In particular,  \begin{align} \label{specific case pert Toepl}  \left\{ \begin{array}{ c l }
 z_{1,1+k}^{(N),a} +z_{N,N-k}^{(N),a}=(N-(k+1))\mu_{a,c},\ &\text{for}\ k\ \text{odd}  \\ 
  z_{1,1+k}^{(N),a} +z_{N,N-k}^{(N),a}=k-N+1,\ &\text{for}\ k\ \text{even}. \end{array} \right. \end{align} This corresponds to the relation of the first row with the last row of the matrix.
  \end{lemma}

 From the above Lemma we conclude that  $z_N$ can be written in the general form
\begin{align} \label{form of z_N}
z_N= -\frac{1}{2} I &+  \sum_{ \substack{k=1 \\ k\ \text{odd}}}^{N-1} \left( z_{N,N-k}^{(N)}( \underline{J}^k - \overline{J}^k) + \sum_{j=k+1}^{N} (N-j) \mu_{a,c}(\overline{\iota}_j + \underline{\iota}_{-j}) \right) \\ &+ \sum_{ \substack{k=1 \\ k\ even}}^{N-1} \left( z_{N,N-k}^{(N)}( \underline{J}^k - \overline{J}^k) - \sum_{j=k+1}^{N} (N-j)(\overline{\iota}_j + \underline{\iota}_{-j}) \right) \notag
\end{align}
where we write $\overline{J}$ for the square matrix with $1$'s in the superdiagonal and $\underline{J}$ for the matrix with $1$'s in the subdiagonal. \\%: For example $$\underline{J_{1}}+ \bar{J}_{-1}= \begin{bmatrix} 0 &-1&0&0&0& \quad \\ 1&0&-1&0&0&\quad \\ 0&1&0&-1&0& \quad   \\ \quad& \quad&\quad&\quad&\quad&\quad \\ \quad& \quad&\quad&\quad&\ddots&\quad \\ 0 & \quad&\quad&\quad&1&0&-1\\ 0 & \quad&\quad&\quad& 0 &1&0  \end{bmatrix}$$
Also $\overline{\iota}_k $ for the matrix with $1$ in the $(k,k+1)$- entry and $\underline{\iota}_{-k}$ for the matrix with $-1$ in the $(k+1,k)$-entry. So for example $$\scalemath{0.86}{ \overline{\iota}_{2}+ \underline{\iota}_{-2} = \begin{bmatrix} 0&0&0&0&\quad \\ 0&0&1&0&\quad \\ 0&-1&0&0&\quad \\ \quad&\quad&\quad&\ddots \\  \quad&\quad&\quad&\quad\\ \quad&\quad&\quad&\quad&0&0   \end{bmatrix}.}$$

%Also we write $\overline{\mathcal{J}}_k $ for the square matrix with $1$'s in the $k$-th off-diagonal and $\underline{\mathcal{J}_k}$ for the mirrored version of this matrix about the cross-diagonal. So for example $$ \scalemath{0.86}{ \overline{\mathcal{J}}_2 + \underline{\mathcal{J}_2} =  \begin{bmatrix} 0  & 1 & 0& \qquad & \quad  \\1 &0 & 0& \qquad & \quad  \\ 0 &0 & \qquad & \quad  \\ \quad & \quad & \Ddots& \quad& \quad \\ \qquad & \quad & \quad &0 &0 \\ \qquad & \quad  &0 &0 & 1\\ \qquad & \quad  & 0 &1 &0   \end{bmatrix}}. $$
 For a visualization: 
  \begin{align*}
  \scalemath{0.69}{  z_N =  \begin{bmatrix} 
 -\frac{1}{2} & -z_{N,N-1}^{(N)}+(N-2) \mu_{a,c} & -z_{N,N-2}^{(N)}-(N-3) &\cdots & -z_{N,2}^{(N)} +\mu_{a,c}& -z_{N,1}^{(N)} \\
 z_{N,N-1}^{(N)}-(N-2)\mu_{a,c} & -\frac{1}{2} &- z_{N,N-1}^{(N)}+(N-3) \mu_{a,c}&\cdots  & -z_{N,3}^{(N)}-1 & z_{1, N-1}^{(N)}  -\mu_{a,c} \\
 z_{N,N-2}^{(N)}+(N-3) &z_{N,N-1}^{(N)} -(N-3)\mu_{a,c} &-\frac{1}{2} & \cdots &-z_{N,4}^{(N)} + \mu_{a,c}& z_{1,N-2}^{(N)}+2\\
  \vdots & \quad &\quad&\quad & \quad  & \vdots \\
   \vdots & \quad &\quad&\ddots & \quad & \quad \\
 \quad &\quad& \quad & \quad & \quad& \quad \\
 \vdots& \quad& \quad & \quad & \quad & \vdots\\
 z_{N,2}^{(N)}-\mu_{a,c}  & -z_{1,N-2}^{(N)}-1 & -z_{1,N-3}^{(N)}+2 \mu_{a,c} &  \cdots &  -\frac{1}{2} & z_{1,2}^{(N)}-(N-2) \mu_{a,c} \\
 -z_{1,N}^{(N)} & -z_{1,N-1}^{(N)}+ \mu_{a,c} &-z_{1,N-2}^{(N)}-2 &  \cdots & -z_{1,2}^{(N)} +(N-2) \mu_{a,c} &  -\frac{1}{2}   \end{bmatrix}}
 \end{align*} 
%where we have defined the $\underline{J}^k, \overline{J}^k, \overline{\iota}_j , \underline{\iota}_{-j}$ in the \hyperref[Notation section]{Notation subsection}.

 %$$ |z_{1k}^{(N)}|\sim O \left(N \frac{ \left( 2\gamma+ \frac{c}{\gamma} \right)^{(N-k)} + \dots+ \left(2\gamma+ \frac{c}{\gamma} \right)^k}{\left(2\gamma+ \frac{c}{\gamma}\right)^N}  \right) \lesssim O  \left( N^2 \right) $$ 
 %so that  $| z_{12}^{(N)}| \sim O  \left( N+1/ ( 2\gamma + \frac{c}{\gamma} )^N \right)$, $|z_{1N}^{(N)} | \sim O \left( N/(2\gamma + \frac{c}{\gamma} )^N+ 1\right)$ .

\begin{proof}[Proof of Lemma \ref{lemma: z_N block}] The proof of this Lemma corresponds in analyzing the final, \hypertarget{mylink}{\textbf{the $N$-th step}} of the matrix equations-sequence. First, from \eqref{eq: 1} we have $$ z_N = z_N^a-\frac{1}{2}I,$$ where $z_N^a$ is antisymmetric matrix. So in order to find the form of $z_N$ we only need to study $z_N^a$ and due to its antisymmetry, we only need to study its upper triagonal part. \\

 We look at the non-bordered entries of the upper triagonal part of \eqref{eq: 3}. That is the equation \label{finalstep} \begin{align} \label{eq: N induction}  c(-\Delta^N z_N^a + z_N^a \Delta^N) -B = \mleft[
\begin{array}{c|c|ccc|c}   \ast & \ast &\ast & \cdots & \ast & \ast \\ 
\hline
\ast&1 &0 & \quad & 0 & \ast \\
\hline
\ast& 0 &1 &\quad &0 &\ast \\
\vdots & \quad& \quad& \ddots & \vdots \\
\hline
\ast & 0 &0 &\quad  & 1& \ast \\
\hline
\ast & \ast & \ast & \cdots & \ast & \ast
\end{array}
\mright] . \end{align} Looking at the diagonal's entries $(i,i)$ for $1<i<N$ of the above equation \eqref{eq: N induction}, we write \begin{align*}
-cz_{i,i-1}^{(N),a}+2cz_{i,i}^{(N),a}-cz_{i,i+1}^{(N),a}+cz_{i-1,i}^{(N),a}-2cz_{i,i}^{(N),a}+cz_{i+1,i}^{(N),a}-(2c+a) =1 
\end{align*}
and using  the antisymmetry of the elements of $z_N^a$, it gives \begin{align*} z_{i,i+1}^{(N),a} &=z_{i-1,i}^{(N),a}- \mu_{a,c} = z_{i-2,i-1}^{(N),a}-2\mu_{a,c} \\=&\dots =z_{1,2}^{(N),a} - (i-1) \mu_{a,c}. \end{align*} 
Therefore, inductively we get \begin{align} \label{eq: k=1 induc} z_{i,i+1}^{(N),a}= z_{1,2}^{(N),a}-(i-1)\mu_{a,c}. \end{align}
At the same time, looking from bottom-right to top-left, we can write \begin{align*} z_{i-1,i}^{(N),a} &=z_{i,i+1}^{(N),a}+ \mu_{a,c} =z_{i+1,i+2}^{(N),a}+ 2\mu_{a,c} \\ =&\dots =z_{N,N-1}^{(N),a}+(i-1) \mu_{a,c}.  \end{align*}  

Then, looking at the super-diagonal's entries, \textit{i.e.} the $(i,i+1)$-entry, for $1< i < N-1$, of equation \eqref{eq: N induction}, we write $$ -cz_{i,i}^{(N),a} + 2cz_{i,i+1}^{(N),a}-cz_{i,i+2}^{(N),a}+cz_{i-1,i+1}^{(N),a}-2cz_{i,i+1}^{(N),a}+cz_{i+1,i+1}^{(N),a}+c=0$$ and that gives $$ z_{i,i+2}^{(N),a}=z_{i-1,i+1}^{(N),a}+1= \dots  = z_{1,3}^{(N),a}+(i-1) $$ and at the same time (reversed direction, \textit{i.e.} from bottom right to top left) $$ z_{i-1,i+1}^{(N),a} =-z_{i+2,i}^{(N),a}-1=\dots=- z_{N,N-2}^{(N),a}- (N-(i+1)) .$$ Similarly, looking at the entries $(i,i+2)$ for $1 < i < N-2$:  \begin{align*} 
  c z_{i-1,i+2}^{(N),a} -2c z_{i,i+2}^{(N),a} + c z_{i+1,i+2}^{(N),a} -c z_{i,i+1}^{(N),a} + 2c z_{i,i+2}^{(N),a} -c z_{i,i+3}^{(N),a}=0.
\end{align*}
Apply \eqref{eq: k=1 induc} twice: $ z_{i+1,i+2}^{(N),a} = z_{1,2}^{(N),a} - i\mu_{a,c}$ and $ -z_{i,i+1}^{(N),a} = -z_{1,2}^{(N),a}+ (i-1)\mu_{a,c}$ and get $$ z_{i-1,i+2}^{(N),a}-\mu_{a,c}= z_{i,i+3}^{(N),a}.  $$ So inductively, \begin{align} \label{eq: k=2 induc}  z_{i,i+3}^{(N),a}= z_{1,4}^{(N),a} - (i-1)\mu_{a,c}. \end{align}  
Also, from the reversed direction we get inductively $$ z_{i,i+3}^{(N),a} = z_{N,N-3}^{(N),a}-(N-3-i).$$ 
For the general case, as stated in the Lemma, we prove it by induction in $k$. For $k=1,2,3$ is true from the above calculations. We do it for $k$ odd. Let it hold for $k-2$, we look at the $(i,i+k-1)$-entry of equation  \eqref{eq: N induction} : for $1<i<N-(k-1)$, \begin{align*} c z_{i-1,i+k-1}^{(N),a}&-2cz_{i,i+k-1}^{(N),a}+cz_{i+1,i+k-1}^{(N),a} -cz_{i,i+(k-2)}^{(N),a} + 2cz_{i,i+k-1}^{(N),a} -cz_{i,i+k}^{(N),a}=0\quad \text{or} \\ z_{i-1,i+k-1}^{(N),a}&- z_{i,i+k}^{(N),a}+ (  z_{i+1, i+1+(k-2)}^{(N),a} - z_{i,i+(k-2)}^{(N),a} )=0 .\end{align*} Then from the induction hypothesis we end up with 
the  \eqref{eq: pert Toeplitz form}. The case $k$ even follows similarly. \\

Now generalize the previous induction formulas for $k$ odd for example and write: 
 $$  z_{i,i+k}^{(N),a}= z_{1,k+1}^{(N),a}-(i-1)\mu_{a,c}$$ and from the reversed direction $$   z_{i,i+k}^{(N),a}=(N-k-i)\mu_{a,c} + z_{N-k,N}^{(N),a}. $$ From these two equations we have the specific case \eqref{specific case pert Toepl}. $k$ even is proven similarly. 
For \eqref{x_ij cross diagonal} we write for $k$ odd: \begin{align*} 
z_{i,i+k}^{(N),a} - z_{N-k-(i-1),N-(i-1)}^{(N),a} &= z_{i-1,i+k-1}^{(N),a}-\mu_{a,c} - (z_{N-k-i,N-i}^{(N),a} + \mu_{a,c} ) \\&= z_{i-1,i+k-1}^{(N),a} - z_{N-k-i,N-i}^{(N),a} - 2\mu_{a,c} \\ &= \cdots = z_{1,k+1}^{(N),a} - z_{N-k,N}^{(N),a}-2(i-1)\mu_{a,c}\\ &= (N-k-2i+1)\mu_{a,c}.
  \end{align*}
  where in the last line we applied \eqref{specific case pert Toepl}. The case $k$ even is proven in the same way.
\end{proof}

The above discussion shows that in order to understand the entries of $z_N$, we need only to understand the vector $\underline{z_N} = (z_{1,2}^{(N)}, z_{1,3}^{(N)},\dots, z_{1,N}^{(N)})$. \\

We state now a Lemma that shows the relation between the elements of $\underline{z_N}$ and the entries of the first row and the last column of  $x_N=[x_{i,j}^{(N)}]$, concluding a relation between $x_{1,j}^{(N)}$ and $x_{i,N}^{(N)}$ about the 'cross diagonal'. 
\begin{lemma}\label{lemma: z relation x} For $3 \le k \le N$,  \begin{align} \label{eq: x relation z}
\left\{ \begin{array}{ c l }
 z_{1,k}^{(N),a} = 1+\frac{\gamma}{c}x_{1,k-1}^{(N)} =- \frac{\gamma}{c}x_{N,N-k+2}^{(N)}-(N-k+1),\ &\text{for}\ k\ \text{odd}  \\ 
  z_{1,k}^{(N),a}= -\mu_{a,c} +\frac{\gamma}{c}x_{1,k-1}^{(N)} =- \frac{\gamma}{c}x_{N,N-k+2}^{(N)}+(N-k+1)\mu_{a,c},\ &\text{for}\ k\ \text{even} \end{array}
 \right. \end{align}
and $z_{1,2}^{(N),a} = \frac{\gamma}{c}x_{1,1}^{(N)} - \frac{T_L+a+2c}{2c}$ and so  for $3 \le k \le N$     \begin{align} \label{eq: x relation cross diagonal}
\left\{ \begin{array}{ c l }
 x_{1,k-1}^{(N)} =- x_{N,N-k+2}^{(N)}-\frac{c}{\gamma}(N-k+2),\ &\text{for}\ k\ \text{odd}  \\ 
  x_{1,k-1}^{(N)} =- x_{N,N-k+2}^{(N)}+\frac{c}{\gamma}(N-k+2) \mu_{a,c},\ &\text{for}\ k\ \text{even}. \end{array}
 \right. \end{align} 
 Also $x_{1,N}^{(N)} = \frac{c}{2\gamma}\mu_{a,c}$, where $\mu_{a,c} := \frac{1+a+2c}{2c}$. 
\end{lemma}
 \begin{proof} We look at the bordered entries of equation \eqref{eq: 3}. Let us first look at $(N,j)$-entry for $j$ even: $$ -c z_{N,j-1}^{(N),a}+2cz_{N,j}^{(N),a}-cz_{N,j+1}^{(N),a}+c z_{N-1,j}^{(N),a} -2c z_{N,j}^{(N),a}= -\gamma x_{N,j}^{(N)}.$$ Using Lemma \ref{lemma: z_N block} we write $$ c  z_{1,N-j+2}^{(N),a} + (j-2)c + c z_{1,N-j }^{(N),a}+jc -cz_{1,N-j}^{(N),a} - (j-1)c = - \gamma x_{N,j}^{(N)}$$ and after the obvious cancellations we have for $j$ even \begin{align} \label{eq: x to z, j even}  x_{N,j}^{(N)} = -\frac{c}{\gamma}z_{1,N-j+2}^{(N),a} - (j-1)\frac{c}{\gamma}. \end{align} Similarly for $j $ odd we have  \begin{align}\label{eq: x to z, j odd} x_{N,j}^{(N)} = -\frac{c}{\gamma}z_{1,N-j+2}^{(N),a} + (j-1)\frac{c}{\gamma}\mu_{a,c}. \end{align}
 Moreover, with exactly the same calculations, but looking at the $(1,j)$-entry of equation \eqref{eq: 3} we get,  for $2 \le j \le N-1$, \begin{align} \label{eq: x to z first line} x_{1,j}^{(N)} = \frac{c}{\gamma}z_{1,j+1}^{(N),a} - \frac{c}{\gamma}\ \text{for}\ j\  \text{even}\quad \text{and}\quad x_{1,j}^{(N)}=  \frac{c}{\gamma}z_{1,j+1}^{(N),a} + \frac{c}{\gamma} \mu_{a,c}\ \text{for}\ j\  \text{odd}.\end{align}
 Now for $k:=N-j+2$ then  $3 \le k \le N$.  Since $N$ is odd, whenever $j$ is odd, $k$ is even and the opposite. Solving the equations \eqref{eq: x to z, j odd} and \eqref{eq: x to z, j even} for $z_{1,k}^{(N),a}$, we get the second equalities in \eqref{eq: x relation z}, whereas solving \eqref{eq: x to z first line} for $\lambda:= j+1$, for $z_{1,\lambda}^{(N),a}$, we get the first equalities in \eqref{eq: x relation z} as well. We conclude with  \eqref{eq: x relation cross diagonal} just by combining the above relations in both cases.\\  Finally to get this specific value for $x_{1,N}^{(N)}$ we look at the $(1,N)$-entry of equation \eqref{eq: 3} and perform the same calculations as above. 
 \end{proof}
 Considering the above Lemma we can write the matrix $z_N$ also as follows: 
 
 \begin{align*}
  \scalemath{0.72}{  z_N =  \begin{bmatrix} 
 -\frac{1}{2} &\frac{\gamma}{c}x_{1,1}^{(N)}-\kappa_L & 1+\frac{\gamma}{c}x_{1,2}^{(N)} &\cdots & -\mu_{a,c}+\frac{\gamma}{c}x_{1,N-2}^{(N)}&1+\frac{\gamma}{c}x_{1,N-1}^{(N)} \\
 -\frac{\gamma}{c}x_{1,1}^{(N)}+\kappa_L & -\frac{1}{2} &  \frac{\gamma}{c}x_{1,1}^{(N)}-\kappa_L-\mu_{a,c}  &\cdots  & \frac{\gamma}{c}x_{1,N-3}^{(N)}+2 & \frac{\gamma}{c}x_{1,N-2}^{(N)}- 2 \mu_{a,c}   \\
%z_{N,N-2}^{(N)}+(N-3) &z_{N,N-1}^{(N)} -(N-3)\mu_{a,c} &-\frac{1}{2} & \cdots &-z_{N,4}^{(N)} + \mu_{a,c}& z_{1,N-2}^{(N)}+2\\
  \vdots & \quad &\quad & \quad  & \vdots \\
   \quad & \quad &\quad & \quad & \quad \\
 \quad &\quad& \quad & \ddots& \quad \\
 %\vdots& \quad& \quad & \quad & \quad & \vdots\\
 \quad & \cdots& \quad& \frac{\gamma}{c}x_{N,N}^{(N)}-\kappa_R-\mu_{a,c}   &  -\frac{1}{2} &  \frac{\gamma}{c}x_{1,1}^{(N)}-\kappa_L-(N-2)\mu_{a,c}   \\
 \quad &\quad&\quad&\cdots & \frac{\gamma}{c}x_{N,N}^{(N)}-\kappa_R &  -\frac{1}{2}   \end{bmatrix}}
 \end{align*}
 where $\kappa_L := \frac{T_L+a+2c}{2c}$ and $\kappa_R := \frac{T_R+a+2c}{2c}$.\\
 
In the following we state a Lemma about the symmetries that hold in $y_N$-block of $b_N$, concluding that all the entries of $y_N$ can be written in terms of the vectors $\underline{y_N}:= (y_{1,N}^{(N)},y_{1,N-1}^{(N)}, \dots, y_{1,1}^{(N)})$ and $\underline{z_N}$. 
\begin{lemma} \label{lemma for y_N}
For $2 \le i \le N-(k+1)$ and $1 \le k \le N-3$, \begin{align}  y_{i-1,i+k}^{(N)} &- y_{i,i+k-1}^{(N)} + (y_{i+1,i+k}^{(N)} - y_{i,i+k+1}^{(N)})=0 \label{symm inside y_N} \\
y_{2,k}^{(N)} &= y_{1,k-1}^{(N)} + y_{1,k+1}^{(N)}+ \frac{\gamma}{c} z_{1,k}^{(N)},\quad \text{for}\quad 2 \le k \le N-1, \label{second line in y_N}\\ 
\text{and}\quad y_{2,N}^{(N)} &= y_{1,N-1}^{(N)} + \frac{2\gamma}{c}z_{1,N}^{(N)} \notag \\ 
y_{k,N}^{(N)}&= \frac{\gamma}{c}(z_{k-1,N}^{(N)} + z_{1,N-(k-2)}^{(N)})+ y_{1,N-(k-1)}^{(N)},\quad \text{for}\quad  2 \le k \le N \label{last column in y_N} 
\end{align}
\end{lemma}
\begin{proof} Due to symmetry of $y_N$ is enough to look at the upper-triagonal part. 
We look at the entries $(i,i+k)$ of equation \eqref{eq: 4a}. For $k=1$ we have $$-y_{i,i}^{(N)}- y_{i,i+2}^{(N)}+y_{i-1,i+1}^{(N)}+y_{i+1,i+1}^{(N)}=0 $$ which is the equation \eqref{symm inside y_N}. For $1<k< N-1$ we prove it by induction in $k$,  like in the proof of Lemma  \eqref{lemma: z_N block}.   Let us now look at the $(1,N)$- entry of \eqref{eq: 4a}: $$ -c y_{1,N-1}^{(N)}+ 2cy_{1,N}^{(N)}-2cy_{1,N}^{(N)} + cy_{2,N}^{(N)} = 2 \gamma z_{1,N}^{(N),a}$$ which gives $y_{2,N}^{(N)} = y_{1,N-1}^{(N)} + \frac{2\gamma}{c}z_{1,N}^{(N)}.$ For \eqref{second line in y_N}  we look at $(1,k)$- entry: $$-cy_{1,k-1}^{(N)}+2cy_{1,k}^{(N)}-cy_{1,k+1}^{(N)} -2cy_{1,k}^{(N)}+cy_{2,k}^{(N)}=\gamma z_{1,k}^{(N),a} $$ which is $$-y_{1,k-1}^{(N)} - y_{1,k+1}^{(N)} +   y_{2,k}^{(N)} = \frac{\gamma}{c} z_{1,k}^{(N),a}$$ and this is the desired equation. For \eqref{last column in y_N}, we look at $(k-1,N)$- entry of \eqref{eq: 4a} for $k\ge 3$. Performing the same calculations as above we get $$  y_{k,N}^{(N)} = \frac{\gamma}{c} z_{k-1,N}^{(N),a} - y_{k-2,N}^{(N)}+y_{k-1,N-1}^{(N)}.$$ Then using the relations \eqref{symm inside y_N} and \eqref{second line in y_N} for each of the terms above, we get the desired relation.  
\end{proof}

With the result of the following Lemma we relate the entries of $\underline{y_N}$ with the entries of $\underline{z_N}$.
\begin{lemma} \label{lemma relation y and z}
Let $B$ be the matrix \eqref{matrix B}. We have \begin{align} 
\underline{y_N} = B^{-1} \underline{\tilde{z}_N}  
\end{align}
where $\underline{\tilde{z}_N}$ is the vector $$  \underline{\tilde{z}_N}=  \begin{bmatrix} \gamma z_{1,N}^{(N)} + \frac{c}{2\gamma}\mu_{a,c} \\
  \frac{c}{\gamma}z_{1,N}^{(N)} - \frac{c}{\gamma}\\ 
  \frac{c}{\gamma}z_{1,N-1}^{(N)} + \frac{c}{\gamma} \mu_{a,c}\\
  \vdots \\
  \frac{c}{\gamma}z_{1,N-i}^{(N)} + \frac{c}{\gamma}\mu_{a,c}\\
  \frac{c}{\gamma}z_{1,N-(i+1)}^{(N)} - \frac{c}{\gamma} \\
  \vdots\\
   \frac{c}{\gamma}z_{1,3}^{(N)} -\frac{c}{\gamma} \\
   \frac{c}{\gamma}z_{1,2}^{(N)} +\frac{T_L+a+2c}{2\gamma}+ \frac{\gamma}{2} 
   \end{bmatrix}
 $$ where $\mu_{a,c} := \frac{1+a+2c}{2c}$. In particular: \begin{align} \label{eq: 1 estimate of y in terms of z} \|  \underline{y_N}\|_2 \lesssim \|\underline{z_N} \|_2 + N^{1/2}.\end{align}
\end{lemma}
\begin{proof} We combine the information for $x_{1i}$'s we get from two equations: first from \eqref{eq: 2}, we remind that equation \eqref{eq: 2} is  $$x_N=By_N+\Gamma z_N$$ and second from the bordered entries of \eqref{eq: 3}, which is $$  -Bz_N+z_NB - B  = J_N^{(\Delta T)}- x_N \Gamma - \Gamma x_N.$$
We look at the element $x_{1,N}^{(N)}$ and we write: \begin{align*}x_{1,N}^{(N)}&=(a+2c)y_{1,N}^{(N)}-cy_{2,N}^{(N)} +\gamma z_{1,N}^{(N),a} = (a+2c)y_{1,N}^{(N)}-cy_{1,N-1}^{(N)}  -2\gamma z_{1,N}^{(N),a}+\gamma z_{1,N}^{(N),a}\\ &=   (a+2c)y_{1,N}^{(N)}-cy_{1,N-1}^{(N)}  -\gamma z_{1,N}^{(N),a} \end{align*} and $$  x_{1,N}^{(N)}= \frac{c}{2\gamma} \mu_{a,c}$$ which give $$  (a+2c)y_{1,N}^{(N)}-cy_{1,N-1}^{(N)} = \gamma z_{1,N}^{(N),a}+ \frac{c}{2\gamma} \mu_{a,c}. $$
Moreover \begin{align*}x_{1,N-1}^{(N)}&= (a+2c) y_{1,N-1}^{(N)} -c y_{2,N-1}^{(N)}+ \gamma z_{1,N-1}^{(N),a} \\&= (a+2c) y_{1,N-1}^{(N)} -cy_{1,N-2}^{(N)}-cy_{1,N}^{(N)} - \gamma z_{1,N-1}^{(N),a} + \gamma z_{1,N-1}^{(N),a} \\ &= (a+2c) y_{1,N-1}^{(N)} -cy_{1,N-2}^{(N)}-cy_{1,N}^{(N)} \end{align*}
and from the proof of Lemma \eqref{lemma: z relation x}, see relation \eqref{eq: x to z first line},  we have  $$ x_{1,N-1}^{(N)} = \frac{c}{\gamma} z_{1,N}^{(N),a}- \frac{c}{\gamma}.$$ Both of them give $$ (a+2c) y_{1,N-1}^{(N)} -cy_{1,N-2}^{(N)}-cy_{1,N}^{(N)} =   \frac{c}{\gamma} z_{1,N}^{(N),a}- \frac{c}{\gamma}. $$ 
In general using again Lemma \ref{lemma for y_N} and relation \eqref{eq: x to z first line}, we have \begin{align*} (a+2c)y_{1,N-i}^{(N)} -cy_{1,N-(i+1)}^{(N)}-cy_{1,N-(i-1)}^{(N)} &=    \left\{ \begin{array}{ c l }
  &\frac{c}{\gamma} z_{1,N-(i-1)}^{(N),a}- \frac{c}{\gamma},\ \text{if}\ i\ \text{odd}  \\ 
 & \frac{c}{\gamma} z_{1,N-(i-1)}^{(N),a}+\frac{c}{\gamma}\mu_{a,c},\ \text{if}\ i\ \text{even}. \end{array}
 \right. \end{align*}
 For $x_{1,1}^{(N)}$ we use that $$   x_{1,1}^{(N)} =\frac{c}{\gamma}z_{1,2}^{(N),a} +\frac{c(T_L+a+2c)}{2\gamma c} $$ from Lemma \eqref{lemma: z relation x}, and from \eqref{eq: 2}, $$ x_{1,1}^{(N)} =(a+2c)y_{1,1}^{(N)}-cy_{1,2}^{(N)}- \frac{\gamma}{2}.$$ Putting the above relations in a more compact form we have $$ B \underline{y_N} = \underline{\tilde{z}_N}.$$ We end up with \eqref{eq: 1 estimate of y in terms of z} considering that $\|B^{-1}\|_2$ is uniformly (in $N$) bounded, since $B$ has bounded spectral gap. 
\end{proof}
\begin{proof}[Proof of Proposition \ref{propos of induction}]
 
The following Lemma shows, through its proof, that there is one unique solution to the Lyapunov matrix equation (since one can explicitly find the entries of $\underline{z_N}$, that determine all the rest) and eventually gives the scaling in $N$ of the entries of $\underline{z_N}$: 
\begin{lemma}\label{lemma for order of z_1j's}
 For $1 \le k \le N-2$, using all the information we have from the block equations in Lemma \eqref{lemm: blocks- equations}, we write all the $z_{1,N-k}^{(N),a}$ in terms of  $z_{1,N}^{(N),a}$, which we then calculate explicitly.  Then, for the order of the entries of $\underline{z_N}$ we have \begin{align}  \label{z_1,N-k and z_1,N}
 \left\{ \begin{array}{ c l }
z_{1,N-k}^{(N),a} &= \mathcal{O}\left( R^kz_{1,N}^{(N),a} + \frac{k}{2}\mu_{a,c}\right),\ \text{for}\ k\ \text{odd}  \\ 
 z_{1,N-k}^{(N),a} &= \mathcal{O}\left( R^kz_{1,N}^{(N),a}-\frac{k}{2} \right),\ \text{for}\ k\ \text{even} \end{array}
 \right. \end{align}
 and $z_{1,N}^{(N),a} = \mathcal{O}\left( R^{1-N} \left( \frac{ \kappa_R-\kappa_L}{2\gamma} \right) \right)$, where $R:= \frac{c}{\gamma^2} + \frac{a+2c}{c}$ and $\mu_{a,c}:= \frac{1+a+2c}{2c}$. Therefore $$ |z_{1,i}^{(N),a}| \lesssim \mathcal{O}\left( (\Delta T) R^{-i+1}  +(N-i) \right),\quad  \text{for}\quad 2\le i \le N$$ where $\Delta T$ is the temperature difference at the ends of the chain.  
\end{lemma}
 \begin{proof} We look at the equations around $x_{k,N}^{(N)}$ for $2 \le k \le N$. First we look at $x_{2,N}^{(N)}$ and from \eqref{eq: x to z, j even} we have $$x_{2,N}^{(N)} = -\frac{c}{\gamma}z_{1,N}^{(N),a}-\frac{c}{\gamma}$$ while from the $(2,N)$-entry of \eqref{eq: 2} we have \begin{align*} x_{2,N}^{(N)}&= -cy_{1,N}^{(N)}+(a+2c)y_{2,N}^{(N)}-cy_{3,N}^{(N)}\\ &=    -cy_{1,N}^{(N)}+(a+2c)y_{1,N-1}^{(N)} + \frac{2\gamma (a+2c)}{c}z_{1,N}^{(N),a}-\gamma(z_{2,N}^{(N),a}+z_{1,N-1}^{(N),a}) -cy_{1,N-2}^{(N)}  \\ &= x_{1,N-1}^{(N)} +  \frac{2\gamma (a+2c)}{c}z_{1,N}^{(N),a} -2\gamma z_{1,N-1}^{(N),a}+ \gamma \mu_{a,c}  \\&= \frac{c}{\gamma}z_{1,N}^{(N),a}-\frac{c}{\gamma}+   \frac{2\gamma (a+2c)}{c}z_{1,N}^{(N),a} -2\gamma z_{1,N-1}^{(N),a}+ \gamma \mu_{a,c}. \end{align*} Combine them and get \begin{align} \label{eq: z_1,n-1 with z_1,n} z_{1,N-1}^{(N),a} = R z_{1,N}^{(N),a}  + \frac{\mu_{a,c}}{2}.\end{align}
 Then we look at $x_{3,N}^{(N)}$: from \eqref{eq: x to z, j odd} we have $$-\frac{c}{\gamma}z_{1,N-1}^{(N),a}+2\frac{c\mu_{a,c}}{\gamma} $$ while from  the $(3,N)$-entry of \eqref{eq: 2} we have similarly \begin{align*} x_{3,N}^{(N)}&= -cy_{2,N}^{(N)}+(a+2c)y_{3,N}^{(N)}-cy_{4,N}^{(N)}\\ &=x_{1,N-2}^{(N)} -2\gamma z_{1,N}^{(N),a}+ \frac{2\gamma (a+2c)}{c}z_{1,N-1}^{(N),a}  -2\gamma z_{1,N-2}^{(N),a} - \frac{\gamma (a+2c)\mu_{a,c}}{c} -2\gamma .\end{align*} Combine them and get  \begin{align*}R z_{1,N-1}^{(N),a} = z_{1,N}^{(N),a}+z_{1,N-2}^{(N),a} + R \frac{\mu_{a,c}}{2} + 1.
 \end{align*} Then considering \eqref{eq: z_1,n-1 with z_1,n} as well, we have \begin{align} z_{1,N-2}^{(N),a}= (R^2-1)z_{1,N}^{(N),a}-1. 
 \end{align}
 In the same manner, but looking around $x_{4,N}^{(N)}$ and $x_{5N}^{(N)}$,  we get \begin{align} z_{1,N-3}^{(N),a} = (R^3-2R) z_{1,N}^{(N),a}+ \frac{3 \mu_{a,c}}{2},\quad  z_{1,N-4}^{(N),a} = (R^4-3R^2+1) z_{1,N}^{(N),a}-2.
 \end{align} respectively. Inductively, we have a  way to write all the elements of $\underline{z_N}$ in terms of $z_{1,N}^{(N),a}$, and looking at the leading order in terms of $N$ we have the general formula \eqref{z_1,N-k and z_1,N} for $1 \le k \le N-2$.
 In particular, for $k=N-3$ (is even by assumption on $N$) and $k=N-2$ (odd) : \begin{align} \label{z_13 and z_14}  z_{1,3}^{(N),a} \sim R^{N-3}z_{1,N}^{(N),a} - \frac{N-3}{2},\quad z_{1,2}^{(N),a} \sim R^{N-2}z_{1,N}^{(N),a} + \frac{(N-2)\mu_{a,c}}{2}. \end{align}
 respectively. Moreover, by looking at $x_{N,N}^{(N)}$ combining \eqref{eq: 2} and \eqref{eq: 3} we have  $$  R z_{1,2}^{(N),a} = R\frac{(N-2)\mu_{a,c}}{2}- \frac{(3-N)}{2}+ \frac{(\kappa_R- \kappa_L)}{2 \gamma} + z_{1,3}^{(N),a}.$$ Plugging in the above equation the relations from \eqref{z_13 and z_14}, we write \begin{align*} (R^{N-1}+R^{N-3} ) z_{1,N}^{(N),a} &+ \frac{R(N-2)\mu_{a,c}}{2} \sim \frac{R(N-2)\mu_{a,c}}{2}  - \frac{(3-N)}{2}  + \frac{(\kappa_R- \kappa_L)}{2 \gamma} - \frac{(N-3)}{2} \\ &\text{which is}\quad  z_{1,N}^{(N),a} \sim R^{1-N} \left( \frac{\kappa_R- \kappa_L}{2 \gamma}\right).
 \end{align*} We conclude the last statement by combining the above estimate on $z_{1,N}^{(N),a}$ with \eqref{z_1,N-k and z_1,N}.
 \end{proof}
%Also, the elements of $y_N$ are symmetric about the 'cross-diagonal' as well. Considering \eqref{choice of z_ij} and the relations among the elements of $y_N$, Lemma \ref{lemma for y_N}:      \begin{align*} & y_{2,k}^{(N)} - y_{N-(k-1),N-1}^{(N)}= \\&= \big( y_{1,k-1}^{(N)} + y_{1,k+1}^{(N)} + \frac{\gamma}{c} z_{1,k}^{(N),a}\big) - \big(  y_{N-k,N}^{(N)} + y_{N-k+2,N}^{(N)} + \frac{\gamma}{c} z_{N-(k-1),N}^{(N),a} \big) =0
%% \end{align*} after \eqref{choice of z_ij} and Lemma \ref{lemma for y_N}. More general we write 
%%% \begin{align*} & y_{i,i+k}^{(N)} - y_{N-k-i+1,N-i+1}^{(N)} = \\ &= \big( y_{1,1+k}^{(N)} + y_{1,3+k}^{(N)}+ \cdots+ y_{1,2i+k-1}^{(N)} + \frac{\gamma}{c}( z_{1,2+k}^{(N),a}+ \cdots+ z_{1,2i+k-2}^{(N),a}) \big) \\ & - \big( y_{N-k,N}^{(N)}  +\cdots+ y_{N-2i-k+2,N}^{(N)} + \frac{\gamma}{c}( z_{N,N-k-1}^{(N),a}+ \cdots+ z_{N,N-2i-k+3}^{(N),a}) \big) = 0 
% \end{align*}
% since from \eqref{last column in y_N}: $y_{N-k,N}^{(N)} = y_{1,k+1}^{(N)}, \dots, y_{1,2i+k-1}^{(N)} = y_{N-2i-k+2,N}^{(N)} $ and from \eqref{choice of z_ij}: $ z_{1,2+k}^{(N),a} = z_{N,N-k-1}^{(N),a}, \dots,   z_{N,N-2i-k+3}^{(N),a} = z_{1,2i+k-2}^{(N),a}$. \\
 
 Now we continue by estimating the entries $\underline{y_N}$: from \eqref{eq: 1 estimate of y in terms of z} and Lemma \ref{lemma for order of z_1j's}, \begin{align*} \|  \underline{y_N}\|_2 &\lesssim \left( \sum_{i=1}^N |z_{1,i}|^2 \right)^{1/2}  + N^{1/2} \lesssim N^{3/2} + N^{1/2} \lesssim N^{3/2}.\end{align*}
This gives that \begin{align} \label{estimate for yij} |y_{1,j}^{(N)}| \lesssim \mathcal{O}(N)  \end{align} and then also, since $ y_{k,N}^{(N)}= \frac{\gamma}{c}(z_{k-1,N}^{(N)} + z_{1,N-(k-2)}^{(N)})+ y_{1,N-(k-1)}^{(N)} $, \begin{align}\label{estimate for y_Nj} |y_{j,N}^{(N)}| \lesssim \mathcal{O}(N). \end{align}

\begin{lemma}[Estimate on the spectral norm of $y_N$] \label{lemma for spectral norm of y_N} For the  spectral norm of $y_N$ we have that $$ \| y_N\|_2 \lesssim \mathcal{O}(N^3). $$ 
 \end{lemma}
 \begin{proof} Let $v=(v_1,v_2,\dots,v_N) \in \mathbb{C}^{N}$. We write $L_i$ for the $i$-th row of the matrix $y_N$ and then calculate \begin{align*} 
 | y_N v |_2^2 &= |L_1 \cdot v |^2 + \cdots + |L_N \cdot v |^2 \\  \leq & N \Bigg( |y_{1,1}^{(N)} v_1|^2 + |y_{1,2}^{(N)}v_2|^2 +\cdots + |y_{1,N}^{(N)} v_{N}|^2 + \qquad \quad (\text{from}\ L_1\cdot v) \\ &\qquad + |y_{1,2}^{(N)}v_2 |^2 +  |y_{2,2}^{(N)}v_2|^2+ \cdots+ | y_{2,N}^{(N)}v_N|^2+ \qquad \quad (\text{from}\ L_2 \cdot v) \\ & \qquad\quad \vdots \\  + |y_{1, \lfloor \frac{N}{2}\rfloor+1}^{(N)}&v_1|^2 +\cdots+  |y_{\lfloor \frac{N}{2}\rfloor+1, \lfloor \frac{N}{2}\rfloor+1}^{(N)}v_{\lfloor \frac{N}{2}\rfloor+1}|^2+ \cdots+  |y_{N, \lfloor \frac{N}{2}\rfloor+1}^{(N)}v_N|^2 + \quad  \big(\text{from}\ L_{\lfloor \frac{N}{2}\rfloor+1} \cdot v \big) \\ & \qquad\quad \vdots \\    & \quad +|y_{1,N}^{(N)} v_1|^2+ | y_{2,N}^{(N)}v_2|^2+ \cdots + |y_{N,N}^{(N)} v_N |^2   \Bigg)  \qquad (\text{from}\ L_N \cdot v )  \end{align*}
 %Considering the symmetry about the 'cross-diagonal' for $y_N$: 
 %\begin{align*} 
  %| y_N v |_2^2  &\le 2N   \Big( |y_{1,1}^{(N)} v_1|^2 + |y_{1,2}^{(N)}v_2|^2 +\cdots + |y_{1,N}^{(N)} v_{N}|^2 + \qquad \quad (\text{from}\ L_1\cdot v) \\ &\qquad + |y_{1,2}^{(N)}v_2 |^2 +  |y_{2,2}^{(N)}v_2|^2+ \cdots+ | y_{2,N}^{(N)}v_N|^2+ \qquad \quad (\text{from}\ L_2 \cdot v) \\ & \qquad\quad \vdots \\   + |y_{1, \lfloor \frac{N}{2}\rfloor +1}^{(N)}&v_1|^2 +\cdots+  |y_{\lfloor \frac{N}{2}\rfloor+1, \lfloor \frac{N}{2}\rfloor+1}^{(N)}v_{\lfloor \frac{N}{2}\rfloor+1}|^2+ \cdots+  |y_{N, \lfloor \frac{N}{2}\rfloor+1}^{(N)}v_N|^2\Big)  \quad (\text{from}\ L_{\lfloor \frac{N}{2}\rfloor+1} \cdot v) .
 %\end{align*}
 We estimate the terms due to the first half of the matrix, \textit{i.e.} the terms until $ L_{\lfloor \frac{N}{2}\rfloor+1} \cdot v $: 
from Lemma \ref{lemma for y_N} we write all the $y_{i,j}^{(N)}$'s in terms of the entries of $\underline{y_N}$ and $\underline{z_N}$ that, due to the observations above, scale at most like $N$. In particular for the second line $$y_{2,k}^{(N)} = y_{1,k-1}^{(N)} + y_{1,k+1}^{(N)} + \frac{\gamma}{c} z_{1,k}^{(N),a}$$ and more general $$ y_{i,i+k}^{(N)}  =  y_{1,1+k}^{(N)} + y_{1,3+k}^{(N)}+ \cdots+ y_{1,2i+k-1}^{(N)} + \frac{\gamma}{c}\left( z_{1,2+k}^{(N),a}+ \cdots+ z_{1,2i+k-2}^{(N),a}\right).$$ 
Then, from \eqref{estimate for yij}:  \begin{align} \label{estimates for y_N}   |L_1 \cdot v |^2 &+ \cdots + \left\vert  L_{\lfloor \frac{N}{2}\rfloor+1} \cdot v \right\vert^2  \lesssim N \Bigg( N^2 |v_1|^2+ \cdots + N^2 |v_{N}|^2+ \\ &\qquad + N^2 |v_1|^2+ 3^2N^2|v_2|^2+ \cdots +3^2N^2|v_{N-1}|^2+ N^2 |v_N|^2 + \notag \\+ N^2|v_1|^2&+3^2N^2|v_2|^2+5^2N^2|v_3|^2+5^2N^2|v_4|^2+ \cdots+5^2N^2|v_{N-2}|^2+3^2|v_{N-1}|^2+ N^2|v_N|^2 \notag \\  & \qquad\quad \vdots \notag  \\ +N^2|v_1|^2&+3^2N^2|v_2|^2+\cdots+ \left(2 \Big\lfloor \frac{N}{2} \Big\rfloor+1 \right)^2N^2 \left\vert v_{\lfloor \frac{N}{2}\rfloor+1}\right\vert^2+\left(2 \Big\lfloor \frac{N}{2} \Big\rfloor-1\right)^2N^2 \left\vert v_{\lfloor \frac{N}{2}\rfloor+2}\right\vert^2+  \notag \\ &\qquad \qquad\qquad\qquad\qquad \qquad\qquad\qquad \cdots+  N^2 |v_N|^2 \Bigg). \notag \end{align}
So the highest order is due to $\left\vert L_{\lfloor \frac{N}{2}\rfloor+1} \cdot v \right\vert^2$ for which we estimate $$\left\vert L_{\lfloor \frac{N}{2}\rfloor+1} \cdot v \right\vert^2 \lesssim   \Bigg( 2N^2\sum_{i=1}^{\lfloor \frac{N}{2}\rfloor+1}(2i-1)^2  \Bigg)   |v|_2^2. $$ The terms $(2i-1)$ in the sum above, denote the number of the entries of $\underline{y_N}, \underline{z_N}$  that each $y_{i,j}^{(N)}$ is given by.\\
Regarding the terms due to the second half of the matrix, we use again Lemma \ref{lemma for y_N}, equations \eqref{symm inside y_N}. This way we write the elements $y_{i,j}^{(N)}$'s in terms of $y_{N,j}^{(N)}$'s and then from
 relation \eqref{last column in y_N}, we have all  the $y_{i,j}^{(N)}$'s in terms of the entries of $\underline{y_N}$ and $\underline{z_N}$, that scale at most like $N$.
  So in the end we have
  $$ | y_N v |_2^2  \lesssim N \Bigg( N^3N^2 \Bigg) | v |_2^2 = N^6 |v|_2^2. $$
 Then $$ \frac{|y_N v |_2}{|v |_2} \lesssim \mathcal{O}(N^3)\qquad \text{and so}\qquad  \|y_N \|_2 \lesssim \mathcal{O}(N^3) .$$ 
 
  Before we finish the proof, we give more details on the estimates \eqref{estimates for y_N} above:\\
   For the first inequality we apply iteratively Lemma \ref{lemma for y_N}. Regarding the row $L_2$:  $$y_{2,2}^{(N)} = y_{1,3}^{(N)} + y_{1,1}^{(N)}+ \frac{\gamma}{c} z_{1,2}^{(N),a}.$$ So $y_{2,2}^{(N)}$ is given by the sum of $3$ terms whose absolute value is of order not more than $\mathcal{O}(N)$. The same holds (from Lemma \ref{lemma for y_N}) for each $y_{2,j}^{(N)}$ for $j \le N-2$, \textit{i.e.} until we reach the 'cross-diagonal'. After the 'cross-diagonal': $y_{2,N}^{(N)}= y_{1,N-1}^{(N)} + \frac{2\gamma}{c}z_{1,N}^{(N)}$,  and $|y_{1,N-1}^{(N)}|, |z_{1,N}^{(N)}|$ have order less than $N$. \\
 Regarding the  row $L_3$: $$y_{3,2}^{(N)} = y_{1,2}^{(N)} + y_{1,4}^{(N)}+ \frac{\gamma}{c}z_{1,3}^{(N),a}$$ is given by the sum of $3$ terms whose absolute value has order less than $N$, while for $y_{3,3}^{(N)}$, by applying Lemma \ref{lemma for y_N} twice, \textit{i.e.} until we end up only with elements of $\underline{y_N}$ and $\underline{z_N}$, we get $$ y_{3,3}^{(N)} = y_{1,3}^{(N)} + y_{1,1}^{(N)}+y_{1,5}^{(N)}+ \frac{\gamma}{c} \left( z_{1,2}^{(N),a} +  z_{1,4}^{(N),a} \right). $$ So $y_{3,3}^{(N)}$ is given by the sum of $5$ terms whose absolute value has order less than $N$.  For $y_{3,j}^{(N)}$, $j \le N-2$ (until the 'cross-diagonal'), apply Lemma \ref{lemma for y_N} twice: the value of $y_{3,j}^{(N)}$ is given by the sum of $5$ such terms, while for $N-1 \le j \le N$, \begin{align*} y_{3,N-1}^{(N)}  &=y_{1,N-3}^{(N)}+ y_{1,N-1}^{(N)} +  \frac{\gamma}{c} z_{1,N-2}^{(N),a} \\ y_{3,N}^{(N)}&=\frac{\gamma}{c}\left(z_{2,N}^{(N)} + z_{1,N-1}^{(N)}\right)+ y_{1,N-2}^{(N)} = \frac{2\gamma}{c}z_{1,N-1}^{(N)}- \frac{\gamma \mu_{a,c}}{c} + y_{1,N-2}^{(N)} \end{align*} 
and so they are given by $3$ terms with absolute value of order at most $N$. 
 %y_{k,N}^{(N)}&= \frac{\gamma}{c}(z_{k-1,N}^{(N)} + z_{1,N-(k-2)}^{(N)})+ y_{1,N-(k-1)}^{(N)}
 
In general, the same holds  for the row $L_i$, $i \le \lfloor \frac{N}{2}\rfloor+1$ from applications of Lemma \ref{lemma for y_N} inductively. For all $y_{i,j}^{(N)}$ we apply Lemma \ref{lemma for y_N} until we have written each $y_{i,j}^{(N)}$ only in terms of entries of $\underline{y_N}$ and $\underline{z_N}$. \\  For $j \le  i $, \textit{i.e.} until the main diagonal, $y_{i,j}^{(N)} $ is given by the sum of $\nu$ terms, whose order is less than $N$, and $$\nu=1,3,5, \cdots,(2i-1)\quad \text{for}\quad y_{i,1}^{(N)}, y_{i,2}^{(N)}, \cdots, y_{i,i}^{(N)},\ \text{respectively}.$$ For that we apply Lemma \ref{lemma for y_N} and write \begin{align*} y_{i,j}^{(N)}= y_{j,i}^{(N)} =   y_{1,i-j+1}^{(N)}+y_{1,i-j+3}^{(N)}+ \cdots+ y_{1,j+i-1}^{(N)}+ \frac{\gamma}{c}\left( z_{1,i-j+2}^{(N),a}+ \cdots+ z_{1,i+j-2}^{(N),a}\right). \end{align*} This formula gives that $y_{i,j}^{(N)}$ is the sum of $(2j-1)$ terms whose absolute value has order less than $\mathcal{O}(N)$.

 The same holds for $j > N-(i-1)$, \textit{i.e.} after the 'cross-diagonal', considering also \eqref{estimate for y_Nj}. As for the rest terms in $L_i$, for $i \le j\le N-(i-1)$:  $y_{i,j}^{(N)}$ is given by the sum of $(2i-1)$ terms whose order is less than $\mathcal{O}(N)$. 
 \end{proof}
 
 \noindent
 Now, from \eqref{eq: 2} we can see that the entries of $x_N$ can be written in terms of entries of $z_N$ as well:
 \begin{align*}  x_{i,j}^{(N)} &= \sum_{\substack{k=1}}^N \beta_{i,k} y_{k,j}^{(N)} + \gamma \sum_{k} ( \delta_{(i=1,k=1)}+\delta_{(i=N,k=N)})  z_{k,j}^{(N)}  \\ &=  \sum_{\substack{k=1,\\ k+j \le N}}^N \beta_{i,k}z_{1,j+k}^{(N)} +  \sum_{\substack{k=1,\\ k+j > N}}^N \beta_{i,k}z_{N,j+k-N-1}^{(N)} + \gamma \sum_{k} ( \delta_{(i=1,k=1)}+\delta_{(i=N,k=N)})  z_{k,j}^{(N)}  
    \end{align*}
    where $\beta_{ij}$ are the elements of the matrix $B$, \eqref{matrix B}, and the entries of $y_N$ are split into two sums regarding their position about the cross diagonal. \\
    %Since for every element, at most $3$ coefficients  $\beta_{ij}$ are non-zero (due to the tridiagonal form of $B$), the order of  $ |x_{i,j}^{(N)}|$ will be $$|x_{i,j}^{(N)}| \lesssim \mathcal{O}(N^3), $$ because of \eqref{final z_Nk}.
   % So, exactly as in Lemma \ref{lemma for spectral norm of y_N} we end up with:  $$\|x_N\|_{2} \lesssim  \mathcal{O}(N^3) . $$
We write   \begin{align*}\|x_N\|_2 \le \|B\|_2 \|y_N\|_2+ \|\Gamma z_N\|_2 \lesssim \|y_N\|_2+ N \lesssim N^3 . \end{align*} 
 
 %Now look at \eqref{eq: y_N} for bounding $\|y_N\|_{\infty}$: $$ \|y_N\|_{\infty} \lesssim O(NN^2+ N^2 ) \sim O(N^3) $$ and so $$ \|y_N\|_2 \le N^{1/2} \|y_N\|_{\infty} \lesssim N^{7/2}.$$ The same is true for $\|x_N\|_2$ and so we end up with \begin{align*} \| b_N\|_2 &\le  \|x_N\|_2 + \|y_N\|_2 \\ & \lesssim O(N^{7/2}).
 %\end{align*}
 \noindent
 Eventually, we write \begin{align*}  \|b_N \|_2 &\leq   \|x_N\|_2 + \|y_N\|_2 \\ & \lesssim \mathcal{O}(N^{3})
 \end{align*} where for the first inequality: since $b_N$ is positive definite, decomposing $b_N$ in its square root matrices: \begin{align*} b_N  = \begin{bmatrix} \chi & \zeta \\ \zeta^T & \psi \end{bmatrix} \begin{bmatrix} \chi & \zeta \\ \zeta^T & \psi \end{bmatrix} &=  \begin{bmatrix} \chi & 0 \\ \zeta^T & 0 \end{bmatrix} \begin{bmatrix} \chi & \zeta \\ 0& 0 \end{bmatrix} +  \begin{bmatrix} 0 & \zeta \\ 0 & \psi \end{bmatrix} \begin{bmatrix} 0 & 0 \\ \zeta^T & \psi \end{bmatrix} \\ &=: X^*X + Y^*Y.
 \end{align*} And since $X^{*}X$ and $XX^{*}$ are unitarily congruent and the same holds for $Y^{*}Y$ and $YY^{*}$ (from polar decomposition for example),   there are unitary matrices $U$, $V \in \mathbb{C}^{N \times N}$  so that:  $$b_N =  X^{*}X+ Y^{*}Y = U XX^{*} U^{*} + V YY^{*}V^{*}  = U  \begin{bmatrix} x_N & 0 \\ 0 & 0 \end{bmatrix} U^{*} + V  \begin{bmatrix} 0 & 0 \\ 0 & y_N \end{bmatrix} V^{*}.
 $$
 Then  it is clear that for the spectral norm (which is unitarily invariant): $$ \left\| \begin{bmatrix} x_N & z_N \\ z_N^T  & y_N \end{bmatrix}  \right\|_2 \leq \|x_N\|_2 + \|y_N\|_2 .$$ 
 \end{proof}
 
 %Also if we take $v=(1,0, \dots, 0)$  and use the lower bound from \eqref{final z_Nk}:  \begin{align*}  \| b_N\|_2^2 &\ge \frac{|b_N v|^2}{|v|^2} \\ &= |x_{1,1}^{(N)}|^2+\dots+|x_{N,1}^{(N)}|^2+|z_{1,1}^{(N)}|^2+ \dots+ |z_{N,1}^{(N)}|^2 \gtrsim N N^2=N^3, \end{align*}

 % so finally $$ N^{3/2} \lesssim\|b_N\|_2 \lesssim N^3 .$$
 
 \subsection{Quantifying the constants from \hyperref[section with functional ineqaualities]{Section 3}.}
 
 Let us first state some facts about the spectrum of the matrix $b_0$ that corresponds to the $0$-th step of the induction:\\
It is known that $b_0$  corresponds to the covariance matrix of the NESS of the harmonic chain (and it has been found explicitly in \cite{RLL67}, see a description of their approach in the beggining of the proof of Lemma \eqref{lemma: z_N block}). From  \cite[Lemma 5.1]{JPS17}, we know that $b_0$ is bounded below and above: $$ T_R \begin{bmatrix} I &0 \\ 0 &B^{-1} \end{bmatrix} \leq b_0 \leq T_L \begin{bmatrix} I &0 \\ 0&B^{-1} \end{bmatrix}. $$

\noindent
So $\| b_0\|_2 $ and $\|b_0^{-1} \|_2$  are uniformly bounded in terms of $N$: note that, Remark \eqref{remark before the induction proof},  $B=-c\ \Delta^N + \sum_{i=1}^N \alpha \delta_i$.  Even though here we will only use that  $\|b_0^{-1} \|_2$ is finite, in fact when $a>0$, $B$ possesses a spectral gap uniformly in $N$ (also in a more general non homogeneous scenario).\\  %We give a short proof:\\

%\begin{lemma}
%\label{prop: Schrodinger op}
%Let $U_{pin}^{(i)}:\mathbb R \rightarrow \mathbb R$ be a family of one body potentials $U_{pin}^{(i)}(q_i):=\alpha_i q_i^2 $ and $U_{int}^{(i)}(q_i-q_j):= c_i (q_i-q_j)^2$ a harmonic interaction potential with $c_i>0$ and $\alpha_i \ge \kappa >0$. Then the matrix $B$ possesses a positive spectral gap uniformly in the number of oscillators $N$, which is at least $\kappa>0.$
%\end{lemma}
%\begin{proof}[Proof of Lemma \ref{prop: Schrodinger op}]
%$B$ can be seen as a family of (Schr\"odinger) operators on $l^2([N])$
%\[  -\Delta^N(c) + \sum_{i=1}^{N-1}\alpha_i \delta_i\]
%where $\delta_i(u)=u_i $ and $c=(c_i)_{i=1}^N $. So we have $ B \geq \kappa>0, $
%since we assumed that $\alpha_i>0$. We are done since Bakry-Emery criterion implies the claim.
%\end{proof}

 We  estimate the equivalence constants in the relation \eqref{equiv T and grad}: $$  C | \nabla f|^2 \leq  \mathcal{T}(f,f) \leq \mathcal{O} \left( N^{3} \right) | \nabla f |^2 , $$ 
 with $C = \| b_0^{-1} \|_2^{-1} >0$ constant that does not depend on $N$: $b_N \ge b_0$ because $ \Pi_N > 2 \tilde{\Theta}$ and so for every $t>0$, $$e^{-t M}\ \Pi_N\ e^{-tM^T} > e^{-t M}\ 2 \tilde{\Theta}\ e^{-t M^T}$$ and since $-M$ is stable (all the characteristic roots have negative real part) we have $$ b_N = \int_0^{\infty} e^{-t M} \Pi_N e^{-tM^T}  dt  >  \int_0^{\infty} e^{-t M} 2 \tilde{\Theta}  e^{-tM^T}  dt = b_0.$$   So $b_N^{-1} \leq b_0^{-1}$ and so $ \| b_{ N}^{-1} \|_2 \leq \| b_0^{-1} \|_2 $ which  is less than a finite constant (because of the spectrum of the discrete Laplacian). So $\| b_{ N}^{-1} \|_2 ^{-1} >0$. \\

\noindent 
Now coming back to the constant $\lambda_N$ given by \eqref{lambda_N} from the proof of the gradient estimate \eqref{T_2>cT}: We remind \eqref{lambda_N}: $$ \lambda_N =  \| b_N \|_{2}^{-1}  - 2(C_{pin}(N) + C_{int}(N) ) \|b_N\|_2 \| b_N^{-1}\|_2. $$ Then \begin{align*}  \lambda_N \gtrsim \frac{1}{N^3} - 2(C_{pin}(N)+C_{int}(N)) N^3  \end{align*}and choosing  \begin{align} \label{behaviour of bounds} C_{pin}(N)+C_{int}(N) = \mathcal{O} \left(C_0\ N^{-6} \right) \end{align} for suitable $C_0 < \frac{1}{2}$, then  $\lambda_N > 0$ and in particular \begin{align} \label{size of lambda_N} \lambda_N \gtrsim (1-2C_0)N^{-3}:= \lambda_0 N^{-3} . \end{align}

Now the constant in the  $\text{LSI(C)}^{h}$ \footnote{Whenever we write $\text{LSI(C)}^{h} $ and $ \text{LSI(C)}^{nh}$ is the constant in the Log-Sobolev inequality in the purely harmonic and weakly anharmonic chain respectively.}, in the harmonic case, is $ \|b_{N} \|_2^2 \| b_{N}^{-1} \|_2 $ and  $ \| b_{ N}^{-1} \|_2 \leq \| b_0^{-1} \|_2 < \infty$. We get the following upper bound  in terms of $N$: $$ \text{LSI(C)}^h\  \lesssim \mathcal{O}( N^{6}) . $$ 
%Also $$\|b_{N} \|_2^2 \| b_{N}^{-1} \|_2 \gtrsim N^3 \| b_N\|_2^{-1} \ge N^3N^{-3} =1$$ and so $$ \mathcal{O}(1) \lesssim  \text{LSI(C)}^h \lesssim \mathcal{O}( N^6) $$ 

 Finally for the $\text{LSI(C)}^{nh}$ in  \textit{weakly anharmonic} case: we remind that the constant is given by $$ \text{LSI(C)}^{nh} = \frac{T_L \|b_N\|_2^2 \|b_N^{-1}\|_2}{2- 2( C_{pin}(N)+C_{int}(N) ) \|b_N\|_2^2\|b_N^{-1}\|_2}$$ as in Proposition \eqref{log sobolev inqlt}.
  Then $$ \text{LSI(C)}^{nh} \lesssim \mathcal{O}(N^6)$$ % and choosing for example the perturbing constants to satisfy $  C_{pin}(N)+C_{int}(N) \sim N^{-7}$, $$ \text{LSI(C)}^{nh} \gtrsim \frac{1}{2-2N^{-6}} \sim \mathcal{O}(1) . $$ 
 So the same bounds hold for this weak anharmonicity as well. \\
 
 \noindent
 \textit{To sum up: for the homogeneous weakly anharmonic chain,  the method described in \hyperref[section with functional ineqaualities]{Section 3} with the perturbed Bakry-Emery criterion,  bounds  the $\text{LSI(C)}$  from above by $N^6$. It gives  a lower bound on the spectral gap that is of order $N^{-3}$ (see the exponential rate in Theorem \ref{maintheorem}). For the purely harmonic chain, since we know that it always decays with $N$ from Proposition \eqref{optimal sg for harmonic}, this lower bound  shows that the spectral gap in this case can not decay at an exponential rate in $N$, it is at most polynomial. }\\
 
 \noindent
In the next Proposition, exploiting the estimates on $\|b_N\|_2$ from the above matrix analysis, we get alternatively the lower bound on the spectral gap of the \textit{harmonic} chain.
%which in fact corresponds to the optimal rate: from the forthcoming work \cite[Proposition 9.1]{BM19} we find exactly the spectral gap of the harmonic homogeneous chain. It is of order $\mathcal{O}(N^{-3})$. %This is done in a different way exploiting properties of the spectrum of the discrete Laplacian, which appears in the blocks of $M$. 

 \begin{proof}[Proof of Proposition \ref{Veselic}]
Remember that $\| b_{ N} \|_2 \lesssim \mathcal{O}(N^3) $ by Proposition \eqref{propos of induction} and that the order of the spectral gap is given by the order of $\inf\{\text{Re}(\mu) : \mu \in \sigma(M) \}$. 
 From  \cite[Inequality (13)]{Ves03}, \cite{GKK91}, we have an estimate for the decay of $e^{-Mt}$: $$ \| e^{-Mt } \|^2 \leq \| b_{N} \| \| b_{N}^{-1} \| e^{-t/ \| b_{N}\| } $$ So, for $u$ be the (normalized) eigenvector corresponding to an eigenvalue of $M$, $\mu>0$, we write  $$ e^{-2 \text{Re}(\mu)  t } = \| e^{-2 \text{Re}(\mu)  t} u \|^2  \leq \|e^{-Mt} u \|^2 \leq \| b_{N} \| \| b_{N}^{-1} \| e^{-t/ \| b_{N}\|} $$
 So we have
 $ -2 \text{Re}(\mu) \leq - \frac{1}{\| b_{N}\|}  $ which means $$  \text{Re}(\mu) \geq  \frac{1}{2 \| b_{N}\|}$$
 and taking the infimum over the real parts of the eigenvalues of $M$, we get that the spectral  gap of the chain is of order bigger than  $N^{-3}$:
 $$\inf\{\text{Re}(\mu) : \mu \in \sigma(M) \} \geq \mathcal{O}\left( \frac{1}{\| b_{N}\|} \right)= \mathcal{O} \left( \frac{1}{N^{3}}\right). $$
 \end{proof}
 
 \noindent
 Eventually, from the whole procedure in this note we have that the spectral gap of the homogeneous harmonic chain is in between $N^{-3}$ and $N^{-1}$. In \cite{BM19} it is proven that this bound is the sharp one. \\
 
 \noindent
  From a simple numerical simulation on the spectral gap of the matrix $M$, the true value is indeed $N^{-3}$. 
 In particular calculating the real part of the smallest eigenvalue of the matrix $M$ in $N$ iterations, and multiplying the result by $N^3$ we get the following behaviour in Figure \ref{matlab fig}, which shows that then the spectral gap converges for large $N$:
 
 \makebox[0pt][l]{
\begin{minipage}{\textwidth}
 \centering
 %\caption[width=0.25\textwidth]{Here $N=300$, the pinning coefficient $a=0$ and the friction constant $\gamma=1$.}
 \includegraphics[width=.6\textwidth]{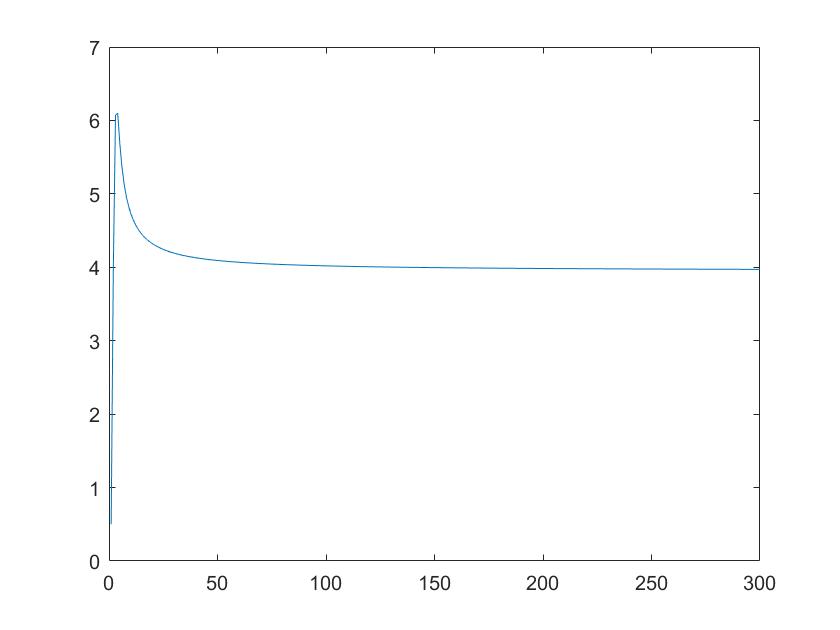}
 \captionof{figure}{Here we implemented it for $N=300$ particles,  the pinning coefficient $a=0$, interaction coefficient $c=1$ and the friction constant $\gamma=1$ .}
 \label{matlab fig}
 \end{minipage}
 }
 
\bigskip
\noindent
{\bf{Acknowledgments}}: I thank my advisor, Cl\'ement Mouhot, for useful conversations, suggestions and encouragement. Also, I would like to thank Simon Becker, for pointing out a mistake in earlier version of this draft and suggesting the reference \cite{Ves03}, to me. Finally, I thank Josephine Evans for  first discussions on the idea. This work was supported by the EPSRC grant EP/L016516/1 for the University of Cambridge CDT, the CCA. 
 
\bibliographystyle{alpha}
\bibliography{bibliography}

\end{document}